\documentclass[12pt]{article}
\usepackage{listings}
\usepackage{color}

\definecolor{dkgreen}{rgb}{0,0.6,0}
\definecolor{gray}{rgb}{0.5,0.5,0.5}
\definecolor{mauve}{rgb}{0.58,0,0.82}

\usepackage{cite}
\usepackage{amsmath,amssymb,amsfonts,mathtools}
\usepackage{algorithmic}
\usepackage{graphicx}
\usepackage{textcomp}
\usepackage{xcolor}
\usepackage{amsthm}
\usepackage{mathabx}
\usepackage{hyperref}
\usepackage{caption}
\usepackage{tablefootnote}
\usepackage{array}
\usepackage{tikz} 

\newcommand{\circled}[1]{%
	\tikz[baseline=(char.base)]{
		\node[shape=circle,draw,inner sep=2pt] (char) {#1};}}

\usepackage{pgf,tikz,pgfplots}
\usepackage{mathrsfs}
\usetikzlibrary{arrows}

\usepackage{xpatch}
\makeatletter
\AtBeginDocument{\xpatchcmd{\@thm}{\thm@headpunct{.}}{\thm@headpunct{}}{}{}}
\makeatother
\usepackage{comment}

\DeclareCaptionType{equ}[][]
\newtheorem{defn}{Definition}
\newtheorem{theorem}{Theorem}[section]
\newtheorem{lemma}[theorem]{Lemma}
\newtheorem{prop}[theorem]{Proposition}
\newtheorem{cor}[theorem]{Corollary}
\newtheorem{example}[theorem]{Example}

\newtheorem{remark}[theorem]{Remark}

\newcommand{\ben}{\begin{equation*}}
	\newcommand{\een}{\end{equation*}}

\newcommand{\F}{\mathbb{F}}

\newcommand{\Spa}{\operatorname{Span}}

\newcommand{\Hull}{\operatorname{Hull}}
\newcommand{\rank}{\operatorname{rank}}
\newcommand{\ev}{\operatorname{ev}}

\begin{document}
	\title{Hull's Parameters of Projective Reed-Muller Codes}
	\author{Yufeng Song \\ School of Mathematics and Statistics\\ Central China Normal University, China \\
		{\tt we72028@gmail.com} \\ \\
		Jinquan Luo \\ School of Mathematics and Statistics\\ Hubei Key Laboratory of Mathematical Sciences \\ Central China Normal University, China \\
		{\tt luojinquan@mail.ccnu.edu.cn} \\
	}
	\date{\today}
	\maketitle
	
	\begin{abstract}
		Projective Reed-Muller codes(PRM codes) are constructed from the family of projective hypersurfaces of a fixed degree over a finite field $\F_q$.  In this paper, we completely determine the minimal distance of the hull of any Projective Reed-Muller codes. Motivated by Nathan Kaplan and Jon-Lark Kim \cite{kaplankim},we extend their results and calculate the hulls' dimension of Projective Reed-Muller Codes in a larger range. We also analyse two special classes of PRM codes apart from self-dual,self-orthgonal and LCD cases, which Kaplan and Kim \cite[section 3]{kaplankim} didn't consider.
	\end{abstract}
	
	
	
	\section{Introduction}
	The Generalized Reed-Muller codes (GRM codes) were first introduced by Kasami, Lin, and Peterson \cite{kasamiLinPeterson}, as well as by Weldon \cite{Weldon}, in 1968. Their research demonstrated that GRM codes are cyclic, allowing for the determination of their minimum distance. Notably, our paper establishes the minimal distance for all projective Reed-Muller codes. Following this, Kasami, Lin, and Peterson \cite{kasamiLinPeterson2} introduced polynomial codes, which extended GRM codes and other code types through a multivariable approach. An in-depth examination of the relationship between multivariable and single-variable approaches to GRM codes was conducted by Delsarte, Goethals, and MacWilliams \cite{DelsarteMGoethalsMacWilliams}. Although these foundational studies were published around 1970, GRM codes continue to attract significant interest in recent research. For instance, \cite{SunDingWang} explored a class of constacyclic codes motivated by punctured generalized Reed-Muller codes.
	
	This paper focuses on projective Reed-Muller codes, which are the projective counterparts of GRM codes. Manin and Vlăduț were among the first to address the natural projective extension of the multivariable GRM construction in their work on algebraic geometric codes \cite{VladutManin}, leading to the development of projective Reed-Muller codes (PRM codes). They derived bounds on the minimum distance of PRM codes over $\mathbb{F}_q$ for order $v$, where $v < q$ and $q$ is a prime power. Lachaud \cite{Lachaud1} formally introduced the term "projective Reed-Muller codes" and investigated PRM codes of orders 1 and 2, further deriving parameters for the case when $v < q$ in \cite{Lachaud2}. Sørensen \cite{Sor} subsequently extended these results by determining parameters for all cases, characterizing the dual of a projective Reed-Muller code, and examining the conditions under which these codes are cyclic. Berger \cite{Ber} studied the automorphism groups associated with these codes. For information on generalized Hamming weights and Hamming weight enumerators for these codes in specific cases, refer to \cite{BeeDatGho2, BeeDatGho, Elk, Kap1}. 
	
	Recently, Ruano and San-José calculated the dimensions of the hulls of projective Reed-Muller codes over the projective plane \cite{RSJ1} and investigated the intersections of pairs of these codes more generally. They also applied their findings to the construction of quantum codes with favorable parameters \cite{RSJ1} and further explored quantum code construction from projective Reed-Muller codes in \cite{RSJ2}. Kaplan and Kim \cite{kaplankim} determined the necessary and sufficient conditions for a PRM code to be self-dual, self-orthogonal, and LCD, and calculated the hull's dimension of a PRM code $PRM(q,m,v)$ for $1 \leq v < q-1$. In our paper, we determine the minimal distance of $PRM(q,m,v)$ for arbitrary values of $q$, $m$, $v$ and calculate the dimension of hull of $PRM(q,m,v)$ in a larger range of $v$ when $q-1 < v < \frac{3(q-1)}{2}$ and $\frac{3(q-1)}{2} < v < 2(q-1)$, $(m-1)(q-1)-\frac{q-1}{2} < v < (m-1)(q-1)$and$(m-2)(q-1) < v < (m-1)(q-1)-\frac{q-1}{2}$. During our study of these codes, we discovered an interesting symmetry in the PRM code as the order $v$ increases from 1 to $\frac{m(q-1)}{2}$.
	
	Combining our result about minimal distance with Kaplan and Kim's result about hull's dimension in \cite[Section 4]{kaplankim}, all hulls $[n_{hull},k_{hull},d_{hull}]$ of PRM codes $PRM(q,m,v)$ with $\mathbb{P}^m(\F_q)$ of dimension 1 and 2  are completely determined. Combining our result about minimal distance with the results in section \ref{hulldim-q-1not-divide} about hull's dimension, all hulls $[n_{hull},k_{hull},d_{hull}]$ of PRM codes $PRM(q,m,v)$ with $\mathbb{P}^m(\F_q)$ of dimension 3 and 4 are completely determined. For $m \ge 5$, we also completely determine the minimal distance $d_{hull}$ and partly determine $k_{hull}$ in a large range. In a word, our results focus not only on projective Reed-Muller codes over the projective line and the projective plane but also take those codes in higher-dimensional projective spaces into consideration. Nonetheless, our paper does not involve hull variation problems or their potential connection with quantum codes, leaving these as a possible direction for future research.
	
	The main aim of this paper is to analyze the hull's parameters of projective Reed-Muller codes. The \emph{hull of a linear code} $C$ is defined as $\Hull(C) = C \cap C^\perp$. There are three significant cases to consider: when $C$ is \emph{self-dual}, meaning $C = C^\perp = \Hull(C)$; when $C$ is \emph{self-orthogonal}, indicating $C = \Hull(C)$; and when $C$ is \emph{LCD (linear code with complementary dual)}, which implies that $\Hull(C)$ is trivial. This concept has been widely studied in the other research \cite{Euclidean and Hermitian LCD MDS codes, New constructions of MDS, combinatorics of LCD codes, Esmaeili, 10, 13, 16}. In addition to the extensive research concerning self-dual, self-orthogonal, and LCD codes (see, for example, \cite{Kim1, Roe1, Massey1, Massey2}), numerous recent research have studied hulls of various code families, often influenced by their connections to quantum error-correcting codes. It is worth to mention that recent work by Gao, Yue, Huang, and Zhang \cite{GaoYueHuaZha}, along with that by Chen, Ling, and Liu \cite{CheLinLiu}, has concentrated on hulls of generalized Reed-Solomon codes.

	The key contributions of our paper are outlined as follows. Analysing two classes of Projective Reed-Muller Code as a complement of Kaplan and Kim's analysis about the sufficient and necessary condition for a Projective Reed-Muller Code to be self-dual,self-orthogonal, and LCD. One case of our analysis is the sufficient and necessary condition for a Projective Reed-Muller Code to be dual-containing. The other case is $v \ge \frac{q-1}{2},v \equiv 0 \pmod{q-1}$ but it isn't any type of self-dual,self-orthgonal,LCD or dual-containing. Combining our results and Kaplan and Kim's results, we get a more detailed characterization of special Projective Reed-Muller Code with special proporsition, such as self-dual,self-orthgonal,LCD or dual-containing. For a Projective Reed-Muller Code determined by three parameters $q,m,v$, Kaplan and Kim compute the dimension of $\operatorname{Hull}(PRM(q,m,v))$ when $1 < v < \frac{q-1}{2}$ , $\frac{q-1}{2} < v < q-1$,$m(q-1)-\frac{q-1}{2} < v < m(q-1)$,$(m-1)(q-1) < v < m(q-1)-\frac{q-1}{2}$. We compute the hull dimension of $PRM(q,m,v)$ when $q-1 < v < \frac{3(q-1)}{2}$ and $\frac{3(q-1)}{2} < v < 2(q-1)$, $(m-1)(q-1)-\frac{q-1}{2} < v < (m-1)(q-1)$and$(m-2)(q-1) < v < (m-1)(q-1)-\frac{q-1}{2}$for a larger range of parameters. In the end, we also completely determine the minimal distance for a Projective Reed-Muller Code $PRM(q,m,v)$ determined by three arbitrary parameters $q,m,v$.Throughout the paper, we provide examples to demonstrate our results.
	
	Our paper is arranged as follows. In the section \ref{pre},we give some basic definitions and proporsitions about projective Reed-Muller codes, the hull of a linear code, and duals of projective Reed-Muller codes. In the section \ref{dual-containing},we determine the the sufficient and necessary condition for a Projective Reed-Muller Code to be dual-containing and point out another special class of Projective Reed-Muller Codes. In section \ref{hulldim-q-1not-divide},we determine the hull dimension of $PRM(q,m,v)$ when $q-1 < v < \frac{3(q-1)}{2}$ and $\frac{3(q-1)}{2} < v < 2(q-1)$, $(m-1)(q-1)-\frac{q-1}{2} < v < (m-1)(q-1)$and$(m-2)(q-1) < v < (m-1)(q-1)-\frac{q-1}{2}$ as the answer for the question proposed by Kaplan and Kim \cite{kaplankim}. In section \ref{minihulldistance}, we determine the minimal distance for any Projective Reed-Muller Code determined by three parameters $q,m,v$. We also give the dimension of the special case in remark \ref{special-case}.

	\section{Preliminaries}\label{pre}
	For a general reference on coding theory we recommend \cite{Pless1, Macwilliams, Pless2}.
	
	\subsection{Projective Reed-Muller Codes}
	First, we introduce the {\em evaluation code}. $\F_q$ is the finite fields with q elements. Let $\F_q^m$ be an affine space over $\F_q$ of dimension m and there are exactly $q^m$ points in it. Hence we can denote $\F_q^m$ by $\F_q^m=\{P_1, \dots , P_{q^m} \}$. We can choose a subset of all these points and denote it by $\mathcal P=\{P_1, \dots , P_n \}$ where $n \le q^m$. Let $\F_q[X_1, \dots, X_m]$ be a polynomials ring consisiting of all polynomials in m variables $X_0,X_1,\dots,X_m$ and let $\mathcal V$ be a subspace of $\F_q[X_1, \dots, X_m]$ with finite dimension. We define the {\em evaluation map} ${\mbox{ev}}_{\mathcal P}$ as a map from $\F_q[X_1, \dots, X_m]$ to $\F_q^n$. For any $F(X) \in \F_q[X_1, \dots, X_m]$, ${\mbox{ev}}_{\mathcal P}(F(X))=(F(P_1), \dots, F(P_n))$. The {\em evaluation code} is the image of $\mathcal{V}$ under the evaluation map.
	
	For convenience, we use the following notations:
	\begin{tabbing}
		$\F_q$ \hspace{40mm} \= Finite fields with q elements, where q is a prime power,$\F_q^{*}=\F_q-\{0\}$.\\
		$\F_q[X_1, \dots, X_m]$ \> A polynomials ring consisiting of all polynomials in m variables  \\
		\>$X_1,\dots,X_m$ with coefficients in $\F_q$.\\
		$\F_q[X_1, \dots, X_m]^v \cup \{0\}$ \> Vectorspace of polynomials in n variables $X_1,\dots,X_m$ with degree $v$.\\
		$\F_q[X_1, \dots, X_m]_h \cup \{0\}$ \> Vectorspace of homogeneous polynomials in m variables.\\
		$\F_q[X_1, \dots, X_m]^v_h \cup \{0\}$ \> Vectorspace of homogeneous polynomials in m variables $X_1,\dots,X_m$. \\
		\> with degree $v$.\\
		$M_{\F_q}[X_1, \dots, X_m]^v$ \> Set of all monomials in n variables $X_1,\dots,X_m$ of degree $v$.\\
		$\mathbb{A}^m$ \> m-dimension affine space over $\F_q$ (=$\F_q^m$).\\
		$\mathbb{P}^m$ \> m-dimension projective space over $\F_q$. \\
		${\pi}_m$\>$\pi_m=q^m+\dots+q+1=\frac{q^{m+1}-1}{q-1}$ is the number of points in $\mathbb{P}^m$.\\
		$\mathcal{P}_m=\{ P_1, \dots ,P_{q^m}\}$ \> The set of all points of $\mathcal{A}^m$. \\
		$\mathcal{P}'_m=\{ P_1', \dots ,P_{{\pi}_m}'\}$ \> The set of all standard representives of the elements in $\mathcal{A}^{m+1}$, \\
		\>i.e the projective points with the form $(0:\dots:0:1:a_j:\dots:a_m)$ ,\\
		\>whose first nonzero element is 1 and all zeros left to the first nonzero \\
		\>element.$\mathcal{P'}_m$ is a subset of $\mathcal{P}_{m+1}$.\\
	\end{tabbing}

With the above notations, we can define the {\em generilized Reed-muller Code} (GRM) and {\em projective Reed-muller Code} (PRM) in a easy way.
	
	\begin{defn}
		A generilized Reed-Muller code with three parameters $q,m,v$ is $GRM(q,m,v)={\mbox{ev}}_{\mathcal P_m}(\F_q[X_1, \dots, X_m]^v \cup \{0\})$, i.e the image of $\F_q[X_1, \dots, X_m]^v \cup \{0\}$ under the evaluation map $\mbox{ev}_{\mathcal{P}_m}$.
	\end{defn}

	\begin{defn}
		Let $\mathbb P^m (\F_q)$ be a projective space of dimension n and $\mathcal P'_m =\{ P'_1, \dots , P'_{{\pi}_m} \}$ be the set of all projective points in $\mathbb P^m (\F_q)$ with standard form, where $\pi_m=q^m+\dots+q+1=\frac{q^{m+1}-1}{q-1}$, the projective Reed-Muller Code with parameters $q,m,v$ is $PRM(q,m,v)={\mbox{ev}}_{\mathcal P'_m}(\F_q[X_0, \dots, X_m]^v_h \cup \{0\})$ ,i.e the image of $\F_q[X_0,X_1, \dots, X_m]^v_h \cup \{0\}$ under the evaluation map $\mbox{ev}_{\mathcal{P'}_m}$.
	\end{defn}
	
 The projective Reed-Muller codes that arise from different orderings of the points of $\mathcal{P}$ and different choices of affine representatives in $\mathcal{P}'$ are monomially equivalent. For convenience, we follow the convention of Lachaud \cite{Lac1, Lac2} and S\o{}rensen \cite{Sor}, where for each projective point we choose the affine representative for which the left-most nonzero coordinate is equal to $1$ in this paper.  Throughout the rest of this paper, whenever we write $PRM(q,m,v)$ we will use this choice of affine representatives.
	
	An easy example is when $v=1$ the projective Reed-Muller code $PRM(q,m,1)$ is a $q$-ary simplex code.  This case is discussed by Lachaud  together with $v=2$ in his initial paper on this topic \cite{Lac1}.Lachaud determined the parameters of the code $PRM(q,m,v)$ in the case where $1 \le v < q$.
	\begin{lemma}{\rm (\cite[Theorem 2]{Lac2})} \label{lem-parameters}
		Assume that $1 \le v < q$. Then the code $PRM(q,m,v)$ has parameters
		{\rm length} $n=\frac{q^{m+1}-1}{q-1}$, {\rm dimension} $k=\binom{m+v}{v}$,
		{\rm distance} $d=(q-v+1)q^{m-1}$.
	\end{lemma}
	
	If $v \ge m(q-1) +1$, then $PRM(q,m,v)$ is trivial, that is, the whole space $\mathbb F_q^n$~\cite[Remark 3]{Sor}.
	For $1 \le v \le m(q-1)$, S\o{}rensen computed the parameters of $PRM(q,m,v)$.  For a recent discussion of S\o{}rensen's computation of the minimum distance and also for a characterization of the minimal weight codewords of $PRM(q,m,v)$, see the paper of Ghorpade and Ludhani \cite{GhoLud}, and also the short note of S\o{}rensen \cite{Sor2}.
	
	\begin{lemma}{\rm(\cite[Theorem 1]{Sor})} \label{lem-parameters-2}
		Assume that $1 \le v \le m(q-1)$. Then the code $PRM(q,m,v)$ has parameters
		\begin{eqnarray*}
			{\mbox{{\rm length} }} n & = &\frac{q^{m+1}-1}{q-1}, \\
			{\mbox{{\rm dimension} }} k &= &\sum_{\substack{t \equiv v\hspace{-.22cm} \pmod{q-1} \\ 0 < t \le v}} \left(
			\sum_{j=0}^{m+1} (-1)^j \binom{m+1}{j} \binom{t-jq+m}{t-jq} \right), \\
			{\mbox{{\rm distance} }} d & = & (q-s)q^{m-r-1},
		\end{eqnarray*}
		where
		\[
		v-1 = r(q-1) +s,~~ 0 \le s < q-1.
		\]
	\end{lemma}
\begin{theorem}{\rm (\cite[Theorem 2]{Sor})} \label{thm-Sor} Let $v$ be an integer satisfying $1 \le v \le m(q-1)$ and let $\ell =m(q-1)-v$. Then
	\begin{enumerate}
		\item [{(i)}] ${PRM(q,m,v)}^{\perp} = PRM(q,m,\ell)$ for $v \not \equiv 0 \pmod{q-1}$.
		\item [{(ii)}] ${PRM(q,m,v)}^{\perp} = {{\Spa}}_{\F_q}  \{ {\bf 1}, PRM(q,m,\ell) \}$ for $v  \equiv 0 \pmod{q-1}$.
	\end{enumerate}
\end{theorem}
\subsection{The Hull of a Linear Code}
Suppose $\vec{x} = (x_1,\ldots, x_n)$ and $\vec{y} = (y_1,\ldots, y_n)$ are elements of $\F_q^n$. Let the Euclidean inner product of $\vec{x}$ and $\vec{y}$ is $\langle x,y \rangle = \sum_{i=1}^n x_i y_i$.  The \emph{dual} of a linear code $C \subseteq \F_q^n$ is
\[
C^\perp = \{y \in \F_q^n \colon \langle y,c\rangle = 0\ \text{ for all } c\in C\}.
\]
The \emph{hull} of a linear code $C$ is defined by $\Hull(C) = C\cap C^\perp$.  A linear code is \emph{self-dual} if $C = C^\perp$ and is \emph{self-orthogonal} if $C \subseteq C^\perp$.A linear code is \emph{dual-containing} if $C \supseteq C^\perp$.  A $k \times n$ matrix $G$ over $\F_q$ is a \emph{generator matrix} for a linear code $C\subset\F_q^n$ of dimension $k$ if the rows of $G$ form a basis for $C$.  We recall a basic fact about the dimension of the hull of a linear code $C$.
\begin{prop}{\rm(\cite[Proposition 3.1]{GueJitGul})} \label{prop-hull}
	Let $C \subseteq \F_q^n$ be a linear code of dimension $k$ and let $G$ be a generator matrix for $C$.  Then $\rank(G G^T) = k - \dim(\Hull(C))$.
	
	In particular, $C$ is self-orthogonal if and only if $GG^T$ is the zero matrix and $C$ is LCD if and only if $GG^T$ is invertible.
\end{prop}

For a Projective Reed-Muller Code determined by three parameters $q,m,v$, Kaplan and Kim compute the dimension of $\operatorname{Hull}(PRM(q,m,v))$ when $1 < v < \frac{q-1}{2}$ and $\frac{q-1}{2} < v < q-1$. By the dual code they got the dimension of $\operatorname{Hull}(PRM(q,m,v))$ when $m(q-1)-\frac{q-1}{2} < v < m(q-1)$,$(m-1)(q-1) < v < m(q-1)-\frac{q-1}{2}$.

\begin{theorem}\label{thm-q_large}\citeonline{kaplankim}
	Suppose that $1 \le v \le \frac{q-1}{2}$. Then
	\[
	\dim(\Hull(PRM(q,m,v)) = \dim (PRM(q,m,v))-1 =\binom{m+v}{v}-1.
	\]
	Moreover, a basis for $\Hull(PRM(q,m,v))$ is given by $\{\ev(f)\}_{f\in \mathcal{M}}$ where $\mathcal{M}$ is the set of monomials of degree $v$ in $x_0,x_1,\ldots, x_m$ except for the monomial $x_m^v$.
\end{theorem}
\begin{cor}\label{cor-q_large}\citeonline{kaplankim}
	Suppose that $m(q-1)- \frac{q-1}{2} < v \le m(q-1)$ and let $\ell = m(q-1) - v$. Then 
	\[
	\dim(\Hull(PRM(q,m,v)))=\dim (PRM(q,m,\ell))-1 =\binom{m+\ell}{\ell}-1.
	\]
\end{cor}
\begin{theorem}\label{thm-k_to_qm1}\citeonline{kaplankim}
	Let $\frac{q-1}{2} < v < q-1$.  Then 
	\[
	\dim(\Hull(PRM(q,m,v))) = \dim(PRM(q,m,v)) - (2v+1 -(q-1)).
	\]
	Moreover, a basis for $\Hull(PRM(q,m,v))$ is given by $\{\ev(f)\}_{f\in \mathcal{M}}$ where $\mathcal{M}$ is the set of monomials of degree $v$ in $x_0,x_1,\ldots, x_m$ except for the monomials $x_{m-1}^{v-a} x_m^a$ where $q-1-v \le a \le v$.
\end{theorem}
\begin{cor}\label{cor-k_to_qm1}\citeonline{kaplankim}
	Suppose that $(m-1)(q-1) < v < m(q-1)- \frac{q-1}{2}$ and let $\ell = m(q-1) - v$. Then 
	\[
	\dim(\Hull(PRM(q,m,v))) = \dim (PRM(q,m,\ell))- (2\ell+1 -(q-1)) =\binom{m+\ell}{\ell}-(2\ell+1 -(q-1)).
	\]
\end{cor}
	\section{dual-containing projective reed-muller code}\label{dual-containing}
	In \cite{kaplankim},Nathan Kaplan and Jon-Lark Kim determine the necessary and sufficent conditions in which a PRM code is self-dual,self-orthgonal and LCD. However, they didn't determine the necessary and sufficent condition for which situation a PRM code is dual-containing i.e $PRM(q,m,v)^{\perp} \subseteq PRM(q,m,v)$.\\
	We give the results of Nathan Kaplan and Jon-Lark \cite{kaplankim} in the following table:
	\begin{center}
		\begin{tabular}{|l|r|}\hline
			\multicolumn{2}{|c|}{Results of Nathan Kaplan and Jon-Lark Kim\citenum{kaplankim}}\\\hline
			Self-dual      & $q,m$ both odd, $v=\frac{m(q-1)}{2}$\\\hline
			Self-orthgonal & $1\le v \le \frac{m(q-1)}{2}$, $2v \equiv 0\pmod {q-1}$\\\hline
			LCD            & $v=m(q-1)$\\\hline
		\end{tabular}
	\end{center}

We give the necessary and sufficent condition in which a PRM code is dual-containing. First, we give the following two lemmas:
\begin{lemma}\label{PRMcontain}
	Let $PRM(q,m,v)$ be a Projective Reed-Muller code and $1 \le v \le m(q-1)$. For any two distinct integer $v_1$,$v_2$ ,if $1 \le v_1 \le v_2 \le m(q-1)$ and $v_1 \equiv v_2 \pmod {q-1}$, then $PRM(q,m,v_1) \subseteq PRM(q,m,v_2)$.
\end{lemma}
\begin{proof}
	For any codeword $\mathbf{c}$ in $PRM(q,m,v_1)$ ,there is an corresponding but no necessarily unique polynomial $F(X)$ in $\F_q[X_0,\dots\,X_m]^{v_1}_h$ such that ${\mbox{ev}}_{\mathcal{P}'_m}(F(X))=\mathbf{c}$. We can write $F(X)=\sum_{i=1}^{t}a_iX_0^{i_0}\cdots X_m^{i_m}$ and let $X_j$ be the indeterminate with the least subscript among $X_0,\dots,X_m$ whose index is nonzero in all terms of $F(X)$. Because $1 \le v_1 \le v_2 \le m(q-1)$ and $v_1 \equiv v_2 \pmod {q-1}$ ,we have that $v_2-v_1=d(q-1)$ for some $d\ge 0$. We can construct a new polynomial $G(X)=X_j^{v_2-v_1}F(X)=X_j^{d(q-1)}F(X)$. $G(X)$ has the same evaluation in each points of $\mathbb{P}^m(\F_q)$ with $F(X)$ and $G(X)$ is of degree $v_2$. Hence $\mathbf{c} \in PRM(q,m,v_2)$. 
\end{proof}
Let $PRM(q,m,v)$ be a Projective Reed-Muller code and $1 \le v \le m(q-1)$. Suppose $v \equiv 0 \pmod{q-1}$, and $\ell=m(q-1)-v$, by \ref{thm-Sor}, ${PRM(q,m,v)}^{\perp} = {{\Spa}}_{\F_q}  \{ {\bf 1}, PRM(q,m,\ell) \}$, in \cite{kaplankim} [Remark 2.8] , Kaplan and Kim prove that $\mathbf{1}\not \in PRM(q,m,\ell)$. Similiarily, we can prove that $\mathbf{1} \not \in PRM(q,m,v)$. Therefore for any $v \equiv 0 \pmod{q-1}$, $\mathbf{1} \not \in PRM(q,m,v)$.
\begin{lemma}\label{1notinPRM}
	For any $v \equiv 0 \pmod{q-1}$, $\mathbf{1} \not \in PRM(q,m,v)$.
\end{lemma}
\begin{proof}
	Because ${PRM(q,m,v)}^{\perp} = {{\Spa}}_{\F_q}  \{ {\bf 1}, PRM(q,m,\ell) \}$, we have that $\mathbf{1}  \in PRM(q,m,v)^{\perp}$. Suppose that $\mathbf{1} \in PRM(q,m,v)$, we have that $\langle\mathbf{1},\mathbf{1}\rangle=0$.But 
	$$\langle\mathbf{1},\mathbf{1}\rangle=1\times n=1\times{\pi_m}=q^m+\cdots+q+1\equiv 1 \pmod{char(\F_q)}$$
	which is a contradiction.
\end{proof}
\begin{theorem}
	Let $PRM(q,m,v)$ be a projective Reed-Muller code and $1 \le v \le m(q-1)$. Then PRM(q,m,v) is dual-containing i.e $PRM(q,m,v)^{\perp} \subseteq PRM(q,m,v)$ if and only if $\frac{m(q-1)}{2} \le v < m(q-1)$ , $2v  \equiv 0 \mod (q-1)$ , $v \not \equiv 0 \mod (q-1)$,i.e $v$ is an odd multiple of $\frac{q-1}{2}$.
\end{theorem}
\begin{proof}
	We prove it in following cases:\\
	\textbf{Case 1:} If $q-1$ doesn't divides $v$ but $\frac{q-1}{2}$ divides $v$, then ${PRM(q,m,v)}^{\perp} =  PRM(q,m,\ell)$ by theorem \ref{thm-Sor}. Because $v \ge \ell$, ${PRM(q,m,v)}^{\perp} \subseteq PRM(q,m,v)$ by lemma \ref{PRMcontain}.\\
	We now prove that $PRM(q,m,v)$ is not dual-containing in other situation. First, we exclude the three cases proved by Kalplan and Kim in which $PRM(q,m,v)$ is self-dual, self-orthgonal or LCD.
	\\
	\textbf{Case 2:} If $v \equiv 0 \pmod{q-1}$,i.e $v$ is an even multiple of $\frac{q-1}{2}$, then ${PRM(q,m,v)}^{\perp} = {{\Spa}}_{\F_q}  \{ {\bf 1}, PRM(q,m,\ell) \}$. By lemma \ref*{1notinPRM}, $\mathbf{1} \not\subset PRM(q,m,v)$. Therefore, ${{\Spa}}_{\F_q}  \{ {\bf 1}, PRM(q,m,\ell) \} \not \subset PRM(q,m,v)$ and $PRM(q,m,v)$ is not dual-containing.\\
	\textbf{Case 3:} If $1 \le v < \frac{m(q-1)}{2}$, then $dim(PRM(q,m,v)) < dim(PRM(q,m,\ell))$ by proposition 2.4 \cite{kaplankim}, which means impossible for $PRM(q,m,v)$ to be a dual-containing code. More briefly, no matter $PRM(q,m,\ell)$ or ${{\Spa}}_{\F_q}  \{ {\bf 1}, PRM(q,m,\ell) \}$ cannot be cotained in $PRM(q,m,v)$.\\
	\textbf{Case 4:} If $\frac{m(q-1)}{2} \le v \le m(q-1)$ and $2k \not \equiv 0 \pmod{q-1}$, then ${PRM(q,m,v)}^{\perp} =  PRM(q,m,\ell)$. From the proof of theorem 3.5 of \cite{kaplankim}, $PRM(q,m,l)$ is not self-orthgonal which is equivalent to the fact that $PRM(q,m,v)$ is not dual-containing.
\end{proof}
\begin{remark}\label{special-case}
	From the above theorem, we find that hull parameters of the \textbf{Case 2}, i.e $v$ is an even multiple of $\frac{q-1}{2}$,is still unknown. It cannot be analysized by the approach for those $v$ can't be divided by $q-1$ in section \ref{hulldim-q-1not-divide}. It also isn't contained in any one of self-dual, self-orthgonal, LCD or dual-containing, we will analyse it in section \ref{minihulldistance} as a special case. 
\end{remark}
	\section{Hull dimension of $PRM$ code when $q-1 < v < 2(q-1)$}\label{hulldim-q-1not-divide}
	For a Projective Reed-Muller Code determined by three parameters $q,m,v$, Kaplan and Kim compute the dimension of $\operatorname{Hull}(PRM(q,m,v))$ when $1 < v < \frac{q-1}{2}$ and $\frac{q-1}{2} < v < q-1$. By the dual code they got the dimension of $\operatorname{Hull}(PRM(q,m,v))$ when $m(q-1)-\frac{q-1}{2} < v < m(q-1)$,$(m-1)(q-1) < v < m(q-1)-\frac{q-1}{2}$. We compute the hull dimension of $PRM(q,m,v)$ in a larger range of parameters when $q-1 < v < \frac{3(q-1)}{2}$ and $\frac{3(q-1)}{2} < v < 2(q-1)$, $(m-1)(q-1)-\frac{q-1}{2} < v < (m-1)(q-1)$ and $(m-2)(q-1) < v < (m-1)(q-1)-\frac{q-1}{2}$ .\\
	\begin{lemma}\label{matrix}
		Let $G$ be an $k \times n$ matrix with rank $k$ and $G'$ be another $k' \times n$ matrixes with the same rank,$k \le k'$, then $rank(GG^T)=rank(G'(G')^T)$.
	\end{lemma}
\begin{proof}
	We can write $G'=PG$, where $P$ is a $k' \times k$ matrix with full column rank. $(G')^T=G^TP^T$, where $P^T$ is a $k \times k'$ matrix with full row rank. Therefore, we have:
	$$
	rank(G'(G')^T)=rank(PGG^TP^T)=rank(GG^TP^T)=rank(GG^T)
	.$$
\end{proof}
By lemma \ref{matrix}, we can analyse the rank of $G'(G')^T$ to determine the hull and don't need to transfer all monomials to the reduced form first, which makes the analysis more convenient.
	\begin{theorem}\label{thm-k_to_3/2(q-1)}
		 Let $m\ge 3, q-1<v<\frac{3(q-1)}{2}$. Then
		$$
		\operatorname{dim}\left(\operatorname{Hull}\left(PRM(q,m,v)\right)\right)=k_1-\left[(3q-v-1)(v-q+2)-q+1\right]; .
		$$
		where $k_1=\operatorname{dim}\left(PRM(q,m,v)\right)=\sum_{\substack{t=v(\bmod q-1) \\
				0<t \leqslant v}}\left(\sum_{j=0}^{m+1}(-1)^j\binom{m+1}{j}\binom{t-j q+m}{t-j q}\right)$\\
		Moreover, a basis for $\operatorname{Hull}\left(PRM(q,m,v)\right)$ is given by $\{\operatorname{ev}(f)\}_{f \in \mathcal{M}}$ where $\mathcal{M}$ is the set of monomials of degree $v$ in $x_0, x_1, \ldots, x_m$ except for the $(3q-v-1)(v-q+2)-q+1$ monomials 
		$x_m^v$ and
		$x_{m-2}^{a_{m-2}}x_{m-1}^{v-a_m-a_{m-2}}x_m^{a_m}$ where 
		$0\leqslant a_{m-2} \leqslant v-(q-1)$,
		$v-(q-1)-a_{m-2} \leqslant a_m \leqslant q-1$
		 or\quad$v-(q-1)+1\leqslant a_{m-2} \leqslant 2v-2(q-1)$,
		$0 \leqslant a_m \leqslant v-a_{m-2}$.
		
	\end{theorem}
	
	\begin{proof}
		Since $q-1<v<\frac{3(q-1)}{2}$, Lemma \ref{lem-parameters-2} implies that $PRM(q,m,v)$ has dimension $k_1=\sum_{\substack{t=v(\bmod q-1) \\
				0<t \leqslant v}}\left(\sum_{j=0}^{m+1}(-1)^j\binom{m+1}{j}\binom{t-j q+m}{t-j q}\right)$. We take the lexicographic ordering of the $k_1$ reduced monomials of degree $v$ in $x_0, \ldots, x_m$ and write this ordered list of monomials as $m_1, \ldots, m_{k_1}$. It is clear that $\operatorname{ev}\left(m_1\right), \ldots, \operatorname{ev}\left(m_{k_1}\right)$ form an ordered basis for $PRM(q,m,v)$. These choices determine a generator matrix $G$. Let $G'$ be the matrix whose rows belongs to  $\operatorname{ev}_{\mathcal{P}'_m}(M_{\F_q}[X_1, \dots, X_m]^v)$, not necessary to be reduced. We can see that all rows of $G$ are the rows of $G'$. We will analyse the nonzero entries of $G' (G')^T$ to determine the rank of $GG^T$. The form of $G' (G')^T$ will make it clear that it has rank $(3q-v-1)(v-q+2)-q+1$, and moreover, will make it clear which monomials of degree $v$ correspond to codewords of $PRM(q,m,v)$ that also lie in ${PRM(q,m,v)}^\perp$.
		
		Suppose $m_i=x_0^{a_0} \cdots x_m^{a_m}$ and $m_j=x_0^{b_0} \cdots x_m^{b_m}$ where $0 \leqslant a_0, \ldots, a_m, b_0, \ldots, b_m \leqslant v$ and $\sum_{\ell=0}^m a_{\ell}=\sum_{\ell=0}^m b_{\ell}=v$. 
		We recall the basic fact that for any positive integer $r$, we have
		\begin{equation}\label{eq-sum}
			\sum_{\beta\in \F_q} \beta^r =
			\begin{cases}
				-1 & \text{if } r \equiv 0 \pmod{q-1} \\
				0 & \text{otherwise}
			\end{cases}.
		\end{equation}
	Our choice of $\mathcal{P}'_m$ implies that the $(i,j)$-entry of $G'(G')^T$ is
	\begin{eqnarray*}
		\sum_{p \in \mathcal{P}'_m} (m_i m_j)(p) & = & \sum_{\alpha_1,\ldots, \alpha_{m} \in \F_q} 1^{a_0+b_0} \alpha_1^{a_1+b_1} \alpha_2^{a_2+b_2}\cdots \alpha_{m}^{a_{m}+b_{m}}   \\
		& & + \sum_{\alpha_2,\ldots, \alpha_{m} \in \F_q} 0^{a_0+b_0} 1^{a_{1}+b_{1}}
		\alpha_2^{a_2+b_2} \cdots \alpha_{m}^{a_{m}+b_{m}}  \\
		& & + \cdots \\
		& & + \sum_{\alpha_m\in \F_q}  \left(0^{a_0+b_0} \cdots 0^{a_{m-2}+b_{m-2}} \right)1^{a_{m-1}+b_{m-1}}  \alpha_m^{a_m+b_m} \\
		& & + \left(0^{a_0+b_0}\cdots 0^{a_{m-1}+b_{m-1}}\right) 1^{a_m+b_m}.
	\end{eqnarray*}
	For any $i$, we have
	\begin{eqnarray*}
		& &  \sum_{\alpha_{i+1},\ldots, \alpha_{m} \in \F_q} \left( 0^{a_{0}+b_{0}} \cdots 0^{a_{i-1}+b_{i-1}}\right) 1^{a_{i}+b_{i}}
		\alpha_{i+1}^{a_{i+1}+b_{i+1}}\cdots \alpha_{m}^{a_{m}+b_{m}}   \\
		& = &  \left( 0^{a_{0}+b_{0}} \cdots 0^{a_{i-1}+b_{i-1}}\right) 1^{a_{i}+b_{i}}
		\left(\sum_{\alpha_{i+1}\in \F_q}  \alpha_{i+1}^{a_{i+1}+b_{i+1}}\right) \cdots \left(\sum_{\alpha_m\in \F_q}  \alpha_m^{a_m+b_m}\right).
	\end{eqnarray*}
	Recall that $0^0 = 1$ denotes the empty product, so if $r$ is a nonnegative integer, then $0^r = 1$ when $r=0$ and $0^r = 0$ otherwise.
	
	We see that the $(i, j)$-entry of $G' (G')^T$ is given by
		$$
		\sum_{p \in \mathcal{P}^{\prime}_m}\left(m_i m_j\right)(p)=\sum_{i=0}^m\left(\left(0^{a_0+b_0} \cdots 0^{a_{i-1}+b_{i-1}}\right) 1^{a_i+b_i}\left(\sum_{\alpha_{i+1} \in \mathbb{F}_q} \alpha_{i+1}^{a_{i+1}+b_{i+1}}\right) \cdots\left(\sum_{\alpha_m \in \mathbb{F}_q} \alpha_m^{a_m+b_m}\right)\right) .
		$$
		
		Equation above implies that $\sum_{\alpha \in \mathbb{F}_q} \alpha^{a+b}=0$ unless $a+b$ is positive and $a+b \equiv 0(\bmod q-1)$. In particular, this is zero unless $a+b \geqslant q-1$. Since $2(q-1)<2 v<3(q-1)$, there is at most two $i$ s for which $a_i+b_i$ is be both positive and divisible by $q-1$. This implies that for any $i \leqslant m-3$,
		$$
		\left(0^{a_0+b_0} \cdots 0^{a_{i-1}+b_{i-1}}\right) 1^{a_i+b_i}\left(\sum_{\alpha_{i+1} \in \mathbb{F}_q} \alpha_{i+1}^{a_{i+1}+b_{i+1}}\right) \cdots\left(\sum_{\alpha_m \in \mathbb{F}_q} \alpha_m^{a_m+b_m}\right)=0 .
		$$
		
		We see that
		$$
		\sum_{p \in \mathcal{P}^{\prime}_m}\left(m_i m_j\right)(p)=\left(0^{a_0+b_0} \cdots 0^{a_{m-3}+b_{m-3}}\right) 1^{a_{m-2}+b_{m-2}}\left(\sum_{\alpha_{m-1} \in \mathbb{F}_q} \alpha_{m-1}^{a_{m-1}+b_{m-1}}\right)\left(\sum_{\alpha_m \in \mathbb{F}_q} \alpha_m^{a_m+b_m}\right)
		$$
		$$+\left(0^{a_0+b_0} \cdots 0^{a_{m-2}+b_{m-2}}\right) 1^{a_{m-1}+b_{m-1}}\left(\sum_{\alpha_m \in \mathbb{F}_q} \alpha_m^{a_m+b_m}\right)+0^{a_0+b_0} \cdots 0^{a_{m-1}+b_{m-1}} 1^{a_m+b_m}
		$$
		The necessary condition for $\sum_{p \in \mathcal{P}^{\prime}_m}\left(m_i m_j\right)(p) \ne 0$ is that $a_0=\cdots a_{m-3}=b_0=\cdots =b_{m-3}=0$,we can classify the value of $a_{m-2},a_{m-1},a_{m},b_{m-2},b_{m-1},b_{m}$ into the following four cases:\\
		\textbf{Case 1}:If $a_{m-2}=a_{m-1}=b_{m-2}=b_{m-1}=0,a_{m}=b_{m}=v$ and the corresonding $m_i$ and $m_j$ is $m_i=x_m^v,m_j=x_m^v$.In this case, $\sum_{p \in \mathcal{P}^{\prime}_m}\left(m_i m_j\right)(p) =1$.\\
		\textbf{Case 2}:If $a_{m-2}+b_{m-2}=0$, $a_{m-1}+b_{m-1}\ne 0$,then 
		$$
		\sum_{p \in \mathcal{P}^{\prime}_m}\left(m_i m_j\right)(p) =\left(\sum_{\alpha_{m-1} \in \mathbb{F}_q} \alpha_{m-1}^{a_{m-1}+b_{m-1}}\right)\left(\sum_{\alpha_n \in \mathbb{F}_q} \alpha_m^{a_m+b_m}\right)+\left(\sum_{\alpha_m \in \mathbb{F}_q} \alpha_m^{a_m+b_m}\right)
		$$
		\textbf{(a)}.If $a_{m-1}+b_{m-1}\equiv 0 \mod(q-1)$,which means $a_{m-1}+b_{m-1}=q-1$ or $a_{m-1}+b_{m-1}=2(q-1)$,then
		$$
		\sum_{p \in \mathcal{P}^{\prime}_m}\left(m_i m_j\right)(p) =-\left(\sum_{\alpha_m \in \mathbb{F}_q} \alpha_m^{a_m+b_m}\right)+\left(\sum_{\alpha_m \in \mathbb{F}_q} \alpha_m^{a_m+b_m}\right)=0
		$$
		In this case, the corresponding $m_i$ and $m_j$ have no influence on our analysis of the matrix $G'(G')^T$.\\ 
		\textbf{(b)}.If $a_{m-1}+b_{m-1}\not\equiv 0 \pmod{q-1}$,then
		$$
		\sum_{p \in \mathcal{P}^{\prime}_m}\left(m_i m_j\right)(p) =\left(\sum_{\alpha_m \in \mathbb{F}_q} \alpha_m^{a_m+b_m}\right)
		$$
		which means the only possible nonzero value of $\sum_{p \in \mathcal{P}^{\prime}_m}\left(m_i m_j\right)(p)$ is $-1$.In order to make it archive this value,we can see that $a_m+b_m$ must be positive and $a_m+b_m \equiv 0 \mod (q-1)$.Consequencely,we have the following possible situations.\\
		\circled{1}$a_m+b_m =q-1$,$a_{m-1}+b_{m-1}=2v-(q-1)$
		,the corresponding $m_i$ and $m_j$ is 
		$$
		m_i=x_{m-1}^{v-a_m}x_m^{a_m}
		$$
		$$
		m_j=x_{m-1}^{v-(q-1)+a_m}x_m^{q-1-a_m}
		$$
		$$
		(0\le a_m \le q-1)
		$$
		In this situation,there exist $q$ nonzero entries -1 in the matrix $G'(G')^T$.\\
		\circled{2}$a_m+b_m =2(q-1)$,$a_{m-1}+b_{m-1}=2v-2(q-1)$
		,the corresponding $m_i$ and $m_j$ is 
		$$
		m_i=x_{m-1}^{v-a_m}x_m^{a_m}
		$$
		$$
		m_j=x_{m-1}^{v-2(q-1)+a_m}x_m^{2(q-1)-a_m}
		$$
		$$
		(2(q-1)-v\le a_m \le v)
		$$
		In this situation,there exist $2v-2(q-1)+1$ nonzero entries -1 in the matrix $G'(G')^T$.\\
		\textbf{Case 3}:If $a_{m-2}+b_{m-2}\ne0$, $a_{m-1}+b_{m-1}= 0$,then 
		$$
		\sum_{p \in \mathcal{P}^{\prime}_m}\left(m_i m_j\right)(p) =\left(\sum_{\alpha_{m-1} \in \mathbb{F}_q} \alpha_{m-1}^{a_{m-1}+b_{m-1}}\right)\left(\sum_{\alpha_m \in \mathbb{F}_q} \alpha_m^{a_m+b_m}\right)=0
		$$
		In this case, the corresponding $m_i$ and $m_j$ have no influence on our analysis of the matrix $G'(G')^T$.\\
		\textbf{Case 4}:If $a_{m-2}+b_{m-2}\ne0$, $a_{m-1}+b_{m-1}\ne 0$,then 
		$$
		\sum_{p \in \mathcal{P}^{\prime}_m}\left(m_i m_j\right)(p) =\left(\sum_{\alpha_{m-1} \in \mathbb{F}_q} \alpha_{m-1}^{a_{m-1}+b_{m-1}}\right)\left(\sum_{\alpha_m \in \mathbb{F}_q} \alpha_m^{a_m+b_m}\right)
		$$
		In order to make the summation above is nonzero,we must have
		$$
		\begin{cases}
			a_{m-2}+b_{m-2}&=2v-2(q-1)\\
			a_{m-1}+b_{m-1}&= q-1\\
			a_m+b_m&=q-1
		\end{cases}
		$$
		The corresponding $m_i$ and $m_j$ are
		$$
		\begin{cases}
			m_i&=x_{m-2}^{a_{m-2}}x_{m-1}^{v-a_m-a_{m-2}}x_m^{a_m}\\
			m_j&=x_{m-2}^{2v-2(q-1)-a_{m-2}}x_{m-1}^{q-1-v+a_m+a_{m-2}}x_m^{q-1-a_m}
		\end{cases}
		$$
		$a_{m-2}$ and $a_m$ can choose the value as follow
		$$
		0\leqslant a_{m-2} \leqslant v-(q-1),v-(q-1)-a_{m-2} \leqslant a_m \leqslant q-1$$
		or
		$$ v-(q-1)+1\leqslant a_{m-2} \leqslant 2v-2(q-1),0 \leqslant a_m \leqslant v-a_{m-2}
		$$
		In this situation,there exist $\left[(3q-v-1)(v-q+2)-q\right]$ nonzero entries 1 in the matrix $GG^T$.\\
		From the analysis of the values and positions of the entries in the matrix $G'(G')^T$,we find that all -1 can
		be eliminated by elementary row operation or elementary column operation.Hence only the entries with value 1 have contributions to the rank of $G'(G')^T$.We claim that $Rank(G'(G')^T)=\left[(3q-v-1)(v-q+2)-q+1\right]$ . We use the Lemma \ref{lem-parameters-2} and Proposition \ref{prop-hull} and get the Theorem 3.1. By exclude the monomials in \textbf{Case 1} and \textbf{Case 4}, a basis for $\operatorname{Hull}\left(PRM(q,m,v)\right)$ is given by $\{\operatorname{ev}(f)\}_{f \in \mathcal{M}}$ where $\mathcal{M}$ is the set of monomials of degree $v$ in $x_0, x_1, \ldots, x_m$ except for the $(3q-v-1)(v-q+2)-q+1$monomials 
		$x_m^v$ and
		$x_{m-2}^{a_{m-2}}x_{m-1}^{v-a_m-a_{m-2}}x_m^{a_m}$ where 
		$0\leqslant a_{m-2} \leqslant v-(q-1)$,
		$v-(q-1)-a_{m-2} \leqslant a_m \leqslant q-1$
		or\quad$v-(q-1)+1\leqslant a_{m-2} \leqslant 2v-2(q-1)$,
		$0 \leqslant a_m \leqslant v-a_{m-2}$.
	\end{proof}

\begin{example}
	We use magma to examine $Hull(PRM(11,3,14))$,the original magma codes is followed.
	\begin{lstlisting}
		 %calculate the dimension of Hull(PRM(11,3,14)) directly
		 q:=11;m:=3;v:=14;
		 F := GF(q);
		 P3 := ProjectiveSpace(F, m);
		 points := Points(P3);
		 R<x_1, x_2, x_3,x_4> := PolynomialRing(F, m+1);
		 polynomials := [R!m : m in SetToSequence(MonomialsOfDegree(R, v)) ];
		 codewords := [];
		 for poly in polynomials do
		 codeword := [];
		 for pt in points do
		 
		 Append(~codeword, Evaluate(poly, Eltseq(pt)));
		 end for;
		 Append(~codewords, Vector(F, codeword));
		 end for;
		 G := Matrix(codewords);
		 linear_code := LinearCode(G);
		 v_mod_q:=v mod (q-1);
		 s_set:={s_1:s_1 in [v_mod_q..v by (q-1)]};
		 dim_prm_qnk:=&+[&+[ (-1)^j*Binomial(m+1, j)*Binomial(t-q*j+m, t-j*q) : j in [0..m+1] ]:t in s_set];
		 dim_prm_qnk-Rank(G*Transpose(G));
		 555
		
		%calculate the dimension of Hull(PRM(11,3,14)) by our formula
		dim_prm_qnk-((3*q-v-1)*(v-q+2)-q+1);
		555 
	\end{lstlisting}
	
\end{example}

\begin{cor}\label{cor-k_to_3/2(q-1)}
	Suppose that $ m\ge 3, m(q-1)- \frac{3(q-1)}{2}< v < (m-1)(q-1)$ and let $\ell = m(q-1) - v$. Then 
	\[
	\dim(\Hull(PRM(q,m,v))) = \dim (\Hull(PRM(q,m,\ell)))=k_2- \left[(3q-\ell-1)(\ell-q+2)-q+1\right] .
	\]
	where $k_2=\operatorname{dim}\left(PRM(q,m,\ell)\right)=\sum_{\substack{t=\ell(\bmod q-1) \\
			0<t \leqslant \ell}}\left(\sum_{j=0}^{m+1}(-1)^j\binom{m+1}{j}\binom{t-j q+m}{t-j q}\right)$
\end{cor}
\begin{example}
	We use magma to examine $Hull(PRM(11,3,17))$,the original magma codes is followed.
	\begin{lstlisting}
		%calculate the dimension of Hull(PRM(11,3,17)) directly
		q:=11;m:=3;v:=17;
		F := GF(q);
		P3 := ProjectiveSpace(F, m);
		points := Points(P3);
		R<x_1, x_2, x_3,x_4> := PolynomialRing(F, m+1);
		polynomials := [R!m : m in SetToSequence(MonomialsOfDegree(R, v)) ];
		codewords := [];
		for poly in polynomials do
		codeword := [];
		for pt in points do
		
		Append(~codeword, Evaluate(poly, Eltseq(pt)));
		end for;
		Append(~codewords, Vector(F, codeword));
		end for;
		G := Matrix(codewords);
		linear_code := LinearCode(G);
		v_mod_q:=v mod (q-1);
		s_set:={s_1:s_1 in [v_mod_q..v by (q-1)]};
		dim_prm_qnk:=&+[&+[ (-1)^j*Binomial(m+1, j)*Binomial(t-q*j+m, t-j*q) : j in [0..m+1] ]:t in s_set];
		dim_prm_qnk-Rank(G*Transpose(G));
		474
		
		%calculate the dimension of Hull(PRM(11,3,17)) by our formula
		l:=m*(q-1)-v;
		l_mod_q:=l mod (q-1);
		t_set:={t_1:t_1 in [l_mod_q..l by (q-1)]};
		dim_prm_qml:=&+[&+[ (-1)^j*Binomial(m+1, j)*Binomial(t-q*j+m, t-j*q) : j in [0..m+1] ]:t in t_set];
		dim_prm_qml-((3*q-l-1)*(l-q+2)-q+1);
		474
		
		
		 
	\end{lstlisting}
	
\end{example}
\begin{proof}
	By the assumption on $v$, Theorem \ref{thm-Sor}(i) implies that ${PRM(q,m,v)}^\perp = PRM(q,m,\ell)$.  We see that $q-1 < \ell < \frac{3(q-1)}{2}$.  Applying Theorem \ref{thm-k_to_qm1} completes the proof.
\end{proof}
 \begin{theorem}\label{th_k_to_2(q-1)}
 	Let $m \ge 4, \frac{3(q-1)}{2}<v<2(q-1)$. Then
 	$$
 	\operatorname{dim}\left(\operatorname{Hull}\left(PRM(q,m,v)\right)\right)=k_1-\left[\sum_{i=3(q-1)-v}^{v}\sum_{j=0}^{3}(-1)^j\binom{3}{j}\binom{i-qj+2}{i-qj}+[2v-3(q-1)+1]\right]; .
 	$$
 	where $k_1=\operatorname{dim}\left(PRM(q,m,v)\right)=\sum_{\substack{t=v(\bmod q-1) \\
 			0<t \leqslant v}}\left(\sum_{j=0}^{m+1}(-1)^j\binom{m+1}{j}\binom{t-j q+m}{t-j q}\right).$\\
 		Moreover,a basis for $\operatorname{Hull}\left(PRM(q,m,v)\right)$ is given by $\{\operatorname{ev}(f)\}_{f \in \mathcal{M}}$ where $\mathcal{M}$ is the set of monomials of degree $v$ in $x_0, x_1, \ldots, x_m$ except for the $\left[\sum_{i=3(q-1)-v}^{v}\sum_{j=0}^{3}(-1)^j\binom{3}{j}\binom{i-qj+2}{i-qj}+[2v-3(q-1)+1]\right]$ polynomials:
 		$$
 		\begin{cases}
 			m_i&=x_{m-1}^{v-a_m}x_m^{a_m},(3(q-1)-v \le a_m \le v)\\
 			m_i&=x_{m-3}^{v-a_{m-2}-a_{m-1}-a_m}x_{m-2}^{a_{m-2}}x_{m-1}^{a_{m-1}}x_m^{a_m},\begin{cases}
 				a_{m-\mu} \le q-1,(\mu=0,1,2)\\
 				3(q-1)-v \le a_{m-2}+a_{m-1}+a_m \le v
 			\end{cases}
 			
 		\end{cases}.
 		$$
 	
 \end{theorem}
\begin{proof}
	Suppose $m_i=x_0^{a_0} \cdots x_m^{a_m}$ and $m_j=x_0^{b_0} \cdots x_m^{b_m}$ where $0 \leqslant a_0, \ldots, a_m, b_0, \ldots, b_m \leqslant v$ and $\sum_{\ell=0}^m a_{\ell}=\sum_{\ell=0}^m b_{\ell}=v$. Following the reasoning from the proof of \ref{thm-k_to_3/2(q-1)}, we see that the $(i, j)$-entry of $G' (G')^T$ is given by
	$$
	\sum_{p \in \mathcal{P}^{\prime}}\left(m_i m_j\right)(p)=\sum_{i=0}^m\left(\left(0^{a_0+b_0} \cdots 0^{a_{i-1}+b_{i-1}}\right) 1^{a_i+b_i}\left(\sum_{\alpha_{i+1} \in \mathbb{F}_q} \alpha_{i+1}^{a_{i+1}+b_{i+1}}\right) \cdots\left(\sum_{\alpha_m \in \mathbb{F}_q} \alpha_m^{a_m+b_m}\right)\right) .
	$$
	
	Equation above implies that $\sum_{\alpha \in \mathbb{F}_q} \alpha^{a+b}=0$ unless $a+b$ is positive and $a+b \equiv 0(\bmod q-1)$. In particular, this is zero unless $a+b \geqslant q-1$. Since $3(q-1)<2 v<4(q-1)$, there is at most three $i$ s for which $a_i+b_i$ is be both positive and divisible by $q-1$. This implies that for any $i \leqslant m-4$,
	$$
	\left(0^{a_0+b_0} \cdots 0^{a_{i-1}+b_{i-1}}\right) 1^{a_i+b_i}\left(\sum_{\alpha_{i+1} \in \mathbb{F}_q} \alpha_{i+1}^{a_{i+1}+b_{i+1}}\right) \cdots\left(\sum_{\alpha_m \in \mathbb{F}_q} \alpha_m^{a_m+b_m}\right)=0 .
	$$
	
	We see that
	\begin{equation}\begin{split}
			&\sum_{p \in \mathcal{P}^{\prime}_m}\left(m_i m_j\right)(p)\\
		   =&\left(0^{a_0+b_0} \cdots 0^{a_{m-4}+b_{m-4}}\right) 1^{a_{m-3}+b_{m-3}}\left(\sum_{\alpha_{m-2} \in \mathbb{F}_q} \alpha_{m-2}^{a_{m-2}+b_{m-2}}\right)\left(\sum_{\alpha_{m-1} \in \mathbb{F}_q} \alpha_{m-1}^{a_{m-1}+b_{m-1}}\right)\left(\sum_{\alpha_m \in \mathbb{F}_q} \alpha_m^{a_m+b_m}\right)\\
		   +& \left(0^{a_0+b_0} \cdots 0^{a_{m-3}+b_{m-3}}\right) 1^{a_{m-2}+b_{m-2}}\left(\sum_{\alpha_{m-1} \in \mathbb{F}_q} \alpha_{m-1}^{a_{m-1}+b_{m-1}}\right)\left(\sum_{\alpha_m \in \mathbb{F}_q} \alpha_m^{a_m+b_m}\right)\\
		   +& \left(0^{a_0+b_0} \cdots 0^{a_{m-2}+b_{m-2}}\right) 1^{a_{m-1}+b_{m-1}}\left(\sum_{\alpha_m \in \mathbb{F}_q} \alpha_m^{a_m+b_m}\right)+0^{a_0+b_0} \cdots 0^{a_{m-1}+b_{m-1}} 1^{a_m+b_m}
		\end{split}		
	\end{equation}
	The necessary condition for $\sum_{p \in \mathcal{P}^{\prime}_m}\left(m_i m_j\right)(p) \ne 0$ is that $a_0=\cdots a_{m-4}=b_0=\cdots =b_{m-4}=0$,we can classify the value of $a_{m-3},a_{m-2},a_{m-1},a_{m},b_{m-3},b_{m-2},b_{m-1},b_{m}$ into the following eight cases:\\
	\textbf{Case 1}:If $a_{m-3}+b_{m-3}=a_{m-2}+b_{m-2}=a_{m-1}+b_{m-1}=0$, i.e $a_{m-3}=b_{m-3}=a_{m-2}=b_{m-2}=a_{m-1}=b_{m-1}=0$.The corresponding $m_i$ and $m_j$ are
	$$
	\begin{cases}
		m_i&=x_m^v\\
		m_j&=x_m^v
	\end{cases}
	$$
	In this case,$\sum_{p \in \mathcal{P}^{\prime}_m}\left(m_i m_j\right)(p)=1$.\\
	\textbf{Case 2}:If $a_{m-3}+b_{m-3}\ne 0,a_{m-2}+b_{m-2}=a_{m-1}+b_{m-1}=0$, i.e 
	
	$$
			\scriptstyle\sum_{p \in \mathcal{P}^{\prime}_m}\left(m_i m_j\right)(p) =\left(0^{a_0+b_0} \cdots 0^{a_{m-4}+b_{m-4}}\right) 1^{a_{m-3}+b_{m-3}}\left(\sum_{\alpha_{m-2} \in \mathbb{F}_q} \alpha_{m-2}^{a_{m-2}+b_{m-2}}\right)\left(\sum_{\alpha_{m-1} \in \mathbb{F}_q} \alpha_{m-1}^{a_{m-1}+b_{m-1}}\right)\left(\sum_{\alpha_m \in \mathbb{F}_q} \alpha_m^{a_m+b_m}\right)\\
		=0
		$$
	In this case, the corresponding $m_i$ and $m_j$ have no influence on our analysis of the matrix $G'(G')^T$.\\
	\textbf{Case 3}:If $a_{m-3}+b_{m-3}= 0,a_{m-2}+b_{m-2}\ne 0,a_{m-1}+b_{m-1}=0$, i.e 
	
	\begin{equation}\begin{split}
			&\sum_{p \in \mathcal{P}^{\prime}_m}\left(m_i m_j\right)(p)\\
			=&\left(0^{a_0+b_0} \cdots 0^{a_{m-4}+b_{m-4}}\right) 1^{a_{m-3}+b_{m-3}}\left(\sum_{\alpha_{m-2} \in \mathbb{F}_q} \alpha_{m-2}^{a_{m-2}+b_{m-2}}\right)\left(\sum_{\alpha_{m-1} \in \mathbb{F}_q} \alpha_{m-1}^{a_{m-1}+b_{m-1}}\right)\left(\sum_{\alpha_m \in \mathbb{F}_q} \alpha_m^{a_m+b_m}\right)\\
			+& \left(0^{a_0+b_0} \cdots 0^{a_{m-3}+b_{m-3}}\right) 1^{a_{m-2}+b_{m-2}}\left(\sum_{\alpha_{m-1} \in \mathbb{F}_q} \alpha_{m-1}^{a_{m-1}+b_{m-1}}\right)\left(\sum_{\alpha_m \in \mathbb{F}_q} \alpha_m^{a_m+b_m}\right)=0
		\end{split}		
	\end{equation}
	In this case, the corresponding $m_i$ and $m_j$ have no influence on our analysis of the matrix $G'(G')^T$.\\
	\textbf{Case 4}:If $a_{m-3}+b_{m-3}= 0,a_{m-2}+b_{m-2}= 0,a_{m-1}+b_{m-1}\ne 0$, i.e 
	
	\begin{equation}\begin{split}
			&\sum_{p \in \mathcal{P}^{\prime}_m}\left(m_i m_j\right)(p)\\
			=&\left(0^{a_0+b_0} \cdots 0^{a_{m-3}+b_{m-3}}\right) 1^{a_{m-2}+b_{m-2}}\left(\sum_{\alpha_{m-1} \in \mathbb{F}_q} \alpha_{m-1}^{a_{m-1}+b_{m-1}}\right)\left(\sum_{\alpha_m \in \mathbb{F}_q} \alpha_m^{a_m+b_m}\right)\\
			+& \left(0^{a_0+b_0} \cdots 0^{a_{m-2}+b_{m-2}}\right) 1^{a_{m-1}+b_{m-1}}\left(\sum_{\alpha_m \in \mathbb{F}_q} \alpha_m^{a_m+b_m}\right)
		\end{split}		
	\end{equation}
	If $a_{m-1}+b_{m-1}\equiv 0\mod(q-1)$,i.e $a_{m-1}+b_{m-1}=q-1,2(q-1)$,we have that$\sum_{p \in \mathcal{P}^{\prime}_m}\left(m_i m_j\right)(p)=0$,which has no influence on our analysis of the matrix $G'(G')^T$. \\
	If $a_{m-1}+b_{m-1}\not\equiv 0\mod(q-1)$,i.e $a_{m-1}+b_{m-1}\ne q-1,2(q-1)$,we have that$\sum_{p \in \mathcal{P}^{\prime}}\left(m_i m_j\right)(p)=\left(0^{a_0+b_0} \cdots 0^{a_{m-2}+b_{m-2}}\right) 1^{a_{m-1}+b_{m-1}}\left(\sum_{\alpha_m \in \mathbb{F}_q} \alpha_m^{a_m+b_m}\right)$.In order to make the summation above is -1,we must have
	$$
	\begin{cases}
		
		a_{m-1}+b_{m-1}&= 2v-2(q-1)\\
		a_m+b_m&=q-1
	\end{cases}
	$$ or
	$$
	\begin{cases}
		
		a_{m-1}+b_{m-1}&= 2v-2(q-1)\\
		a_m+b_m&=2(q-1)
	\end{cases}
	$$ or
	$$
	\begin{cases}
		
		a_{m-1}+b_{m-1}&= 2v-3(q-1)\\
		a_m+b_m&=2(q-1)
	\end{cases}
	$$
	The corresponding $m_i$ and $m_j$ are
	$$
	\begin{cases}
		m_i&=x_{m-1}^{v-a_m}x_m^{a_m}\\
		m_j&=x_{m-1}^{v-2(q-1)+a_m}x_m^{q-1-a_m}\\
		&(0\le a_m \le q-1)
	\end{cases}
	$$ or
	$$
	\begin{cases}
		m_i&=x_{m-1}^{v-a_m}x_m^{a_m}\\
		m_j&=x_{m-1}^{v-2(q-1)+a_m}x_m^{2(q-1)-a_m}\\
		&(2(q-1)-v \le a_m \le v)
	\end{cases}
	$$ or
	$$
	\begin{cases}
		m_i&=x_{m-1}^{v-a_m}x_m^{a_m}\\
		m_j&=x_{m-1}^{v-3(q-1)+a_m}x_m^{2(q-1)-a_m}\\
		&(3(q-1)-v \le a_m \le v)
	\end{cases}
	$$. \\
	Hence there are totally $q+[2v-2(q-1)+1]+[2v-3(q-1)+1]=4v-4q+7$ entries in $G'(G')^T$ has value -1.\\
	\textbf{Case 5}:If $a_{m-3}+b_{m-3}\ne 0,a_{m-2}+b_{m-2}\ne 0,a_{m-1}+b_{m-1}= 0$, i.e $\sum_{p \in \mathcal{P}^{\prime}_m}\left(m_i m_j\right)(p)=0^{a_0+b_0} \cdots 0^{a_{m-1}+b_{m-1}} 1^{a_m+b_m}$,and in order to make it nonzero,we have following possible choices:\\
	\textbf{(a)}
	$$
	\begin{cases}
		a_m=0,b_m= 0\\
		a_{m-3}+b_{m-3}\ne 0,a_{m-2}+b_{m-2}\ne 0
	\end{cases}
	$$
	There exists $(v+1)^2$ pairs of $m_i$ and $m_j$ 
	$$
	\begin{cases}
		m_i&=x_{m-3}^{v-a_{m-2}}x_{m-2}^{a_{m-2}},(0\le a_{m-2}\le v)\\
		m_j&=x_{m-3}^{v-b_{m-2}}x_{m-2}^{b_{m-2}},(0\le b_{m-2}\le v)
	\end{cases}
	$$
	because
	$$\begin{cases}
		[v-a_{m-2}]+[v-b_{m-2}] \ne 0\\
		a_{m-2}+b_{m-2}\ne 0
	\end{cases}$$	
	Hence the choice for $m_i$ and $m_j$ must eliminate
	$$\begin{cases}
		m_i&= x_{m-2}^v\\
		m_j&= x_{m-2}^v\\
	\end{cases}	$$and
    $$\begin{cases}
	m_i&= x_{m-3}^v\\
	m_j&= x_{m-3}^v\\
    \end{cases}	$$
    Therefore,only $[(v+1)^2-2]$ pairs $m_i$ and $m_j$ satisfy $\sum_{p \in \mathcal{P}^{\prime}_m}\left(m_i m_j\right)(p)=1$.\\
	\textbf{(b)}
	$$
	\begin{cases}
		a_m=0,b_m\ne 0\\
		a_{m-3}+b_{m-3}\ne 0,a_{m-2}+b_{m-2}\ne 0
	\end{cases}
    $$
    There exists $\frac{v(v+1)^2}{2}$ pairs of $m_i$ and $m_j$
    $$
    \begin{cases}
    	m_i&=x_{m-3}^{v-a_{m-2}}x_{m-2}^{a_{m-2}}\\
    	m_j&=x_{m-3}^{v-b_{m-2}-{b_m}}x_{m-2}^{b_{m-2}}x_m^{b_m}
    \end{cases}
    $$
    because
    $$\begin{cases}
    	[v-a_{m-2}]+[v-b_{m-2}-{b_m}] \ne 0\\
    	a_{m-2}+b_{m-2}\ne 0
    \end{cases}$$	
    Hence the choice for $m_i$ and $m_j$ must eliminate the $v$ pairs
    $$\begin{cases}
    	m_i&= x_{m-2}^v\\
    	m_j&= x_{m-2}^{v-b_m}x_m^{b_m},(1 \le b_m \le v)\\
    \end{cases}	$$
    Therefore,only $\frac{v(v+1)^2}{2}-v$ pairs $m_i$ and $m_j$ satisfy $\sum_{p \in \mathcal{P}^{\prime}}\left(m_i m_j\right)(p)=1$.\\
    \textbf{(c)}$$
    \begin{cases}
    	a_m\ne 0,b_m= 0\\
    	a_{m-3}+b_{m-3}\ne 0,a_{m-2}+b_{m-2}\ne 0
    \end{cases}
    $$
    Similiar to (b),we can get that there exists $\frac{v(v+1)^2}{2}$ pairs of $m_i$ and $m_j$
    $$
    \begin{cases}
    	m_i&=x_{m-3}^{v-a_{m-2}-{a_m}}x_{m-2}^{a_{m-2}}x_m^{a_m}\\
    	m_j&=x_{m-3}^{v-b_{m-2}}x_{m-2}^{b_{m-2}}
    \end{cases}
    $$
    because
    $$\begin{cases}
    	[v-b_{m-2}]+[v-a_{m-2}-{a_m}] \ne 0\\
    	a_{m-2}+b_{m-2}\ne 0
    \end{cases}$$	
    Hence the choice for $m_i$ and $m_j$ must eliminate the $v$ pairs
    $$\begin{cases}
    	m_i&= x_{m-2}^{v-a_m}x_m^{a_m},(1 \le a_m \le v)\\
    	m_j&= x_{m-2}^v\\
    \end{cases}	$$
    Therefore,only $\frac{v(v+1)^2}{2}-v$ pairs $m_i$ and $m_j$ satisfy $\sum_{p \in \mathcal{P}^{\prime}_m}\left(m_i m_j\right)(p)=1$.\\
    \textbf{(d)}$$
    \begin{cases}
    	a_m\ne 0,b_m\ne 0\\
    	a_{m-3}+b_{m-3}\ne 0,a_{m-2}+b_{m-2}\ne 0
    \end{cases}
    $$
    Similiar to (c),we can get that there exists $\frac{v^2(v+1)^2}{4}$ pairs of $m_i$ and $m_j$
    $$
    \begin{cases}
    	m_i&=x_{m-3}^{v-a_{m-2}-{a_m}}x_{m-2}^{a_{m-2}}x_m^{a_m}\\
    	m_j&=x_{m-3}^{v-b_{m-2}-b_m}x_{m-2}^{b_{m-2}}x_m^{b_m}
    \end{cases}
    $$
    because
    $$\begin{cases}
    	[v-a_{m-2}-{a_m}]+[v-b_{m-2}-b_m] \ne 0\\
    	a_{m-2}+b_{m-2}\ne 0
    \end{cases}$$	
    Hence the choice for $m_i$ and $m_j$ must eliminate the $v^2$ pairs
    $$\begin{cases}
    	m_i&= x_{m-2}^{v-a_m}x_m^{a_m},(1 \le a_m \le v)\\
    	m_j&= x_{m-2}^{v-b_m}x_m^{b_m},(1 \le b_m \le v)
    \end{cases}	$$
    Therefore,only $\frac{v^2(v+1)^2}{4}-v^2$ pairs $m_i$ and $m_j$ satisfy $\sum_{p \in \mathcal{P}^{\prime}_m}\left(m_i m_j\right)(p)=1$.\\
    From (a),(b),(c),(d),we can conclude that there exactly $(v+1)^2(v+\frac{v^2}{4})-1$ entries with value 1 in Case 5.\\
    \textbf{Case 6}:If $a_{m-3}+b_{m-3}\ne 0,a_{m-2}+b_{m-2}= 0,a_{m-1}+b_{m-1}\ne 0$,i.e $\sum_{p \in \mathcal{P}^{\prime}_m}\left(m_i m_j\right)(p)=0^{a_0+b_0} \cdots 0^{a_{m-1}+b_{m-1}} 1^{a_m+b_m}$.Similiar to Case 5,we only need to replace all $a_{m-2}$ into $a_{m-1}$ and complete the proof.There exactly $(v+1)^2(v+\frac{v^2}{4})-1$ entries with value 1 in Case 6.\\
    \textbf{Case 7}:If $a_{m-3}+b_{m-3}= 0,a_{m-2}+b_{m-2}\ne 0,a_{m-1}+b_{m-1}\ne 0$,i.e \begin{equation}\begin{split}
    		&\sum_{p \in \mathcal{P}^{\prime}_m}\left(m_i m_j\right)(p)\\
    		=&\left(0^{a_0+b_0} \cdots 0^{a_{m-4}+b_{m-4}}\right) 1^{a_{m-3}+b_{m-3}}\left(\sum_{\alpha_{m-2} \in \mathbb{F}_q} \alpha_{m-2}^{a_{m-2}+b_{m-2}}\right)\left(\sum_{\alpha_{m-1} \in \mathbb{F}_q} \alpha_{m-1}^{a_{m-1}+b_{m-1}}\right)\left(\sum_{\alpha_m \in \mathbb{F}_q} \alpha_m^{a_m+b_m}\right)\\
    		+& \left(0^{a_0+b_0} \cdots 0^{a_{m-3}+b_{m-3}}\right) 1^{a_{m-2}+b_{m-2}}\left(\sum_{\alpha_{m-1} \in \mathbb{F}_q} \alpha_{m-1}^{a_{m-1}+b_{m-1}}\right)\left(\sum_{\alpha_m \in \mathbb{F}_q} \alpha_m^{a_m+b_m}\right)\end{split}		
    \end{equation}
    If $a_{m-2}+b_{m-2}\equiv 0 \mod(q-1)$,then $\sum_{p \in \mathcal{P}^{\prime}_m}\left(m_i m_j\right)(p)=0$.
    \\Hence $a_{m-2}+b_{m-2}\ne q-1,2(q-1),3(q-1)$ and
    $$
    	\sum_{p \in \mathcal{P}^{\prime}_m}\left(m_i m_j\right)(p)=\left(0^{a_0+b_0} \cdots 0^{a_{m-3}+b_{m-3}}\right) 1^{a_{m-2}+b_{m-2}}\left(\sum_{\alpha_{m-1} \in \mathbb{F}_q} \alpha_{m-1}^{a_{m-1}+b_{m-1}}\right)\left(\sum_{\alpha_m \in \mathbb{F}_q} \alpha_m^{a_m+b_m}\right)
    $$
    In order to make it choose nonzero value 1,we must have\\
    $$\begin{cases}
    	a_{m-2}+b_{m-2}\ne q-1,2(q-1),3(q-1)\\
    	a_{m-1}+b_{m-1}\equiv 0 \mod (q-1)\\
    	a_{m}+b_{m}\equiv 0 \mod (q-1)
    \end{cases}$$
    We can state it in the following three cases:\\
    \textbf{(a)}
    $$\begin{cases}
    	0<a_{m-2}+b_{m-2}< (q-1),(q-1)<a_{m-2}+b_{m-2}< 2(q-1)\\
    	a_{m-1}+b_{m-1}=q-1\\
    	a_{m}+b_{m}=q-1
    \end{cases}$$
    There exists $\frac{q(3q-1)}{2}$ pairs corresponding $m_i$ and $m_j$ 
    $$
    \begin{cases}
    	m_i&=x_{m-2}^{v-a_{m-1}-{a_m}}x_{m-1}^{a_{m-1}}x_m^{a_m}\\
    	m_j&=x_{m-2}^{v-2(q-1)+a_{m-1}+a_m}x_{m-1}^{q-1-a_{m-1}}x_m^{q-1-a_m}
    \end{cases}
    $$
    \textbf{(b)}
    $$\begin{cases}
    	0<a_{m-2}+b_{m-2}< (q-1)\\
    	a_{m-1}+b_{m-1}=2(q-1)\\
    	a_{m}+b_{m}=q-1
    \end{cases}$$
    There exists $\frac{q(5q-3)}{2}$ pairs corresponding $m_i$ and $m_j$ 
    $$
    \begin{cases}
    	m_i&=x_{m-2}^{v-a_{m-1}-{a_m}}x_{m-1}^{a_{m-1}}x_m^{a_m}\\
    	m_j&=x_{m-2}^{v-3(q-1)+a_{m-1}+a_m}x_{m-1}^{2(q-1)-a_{m-1}}x_m^{q-1-a_m}
    \end{cases}
    $$
    \textbf{(c)}
    $$\begin{cases}
    	0<a_{m-2}+b_{m-2}< (q-1)\\
    	a_{m-1}+b_{m-1}=(q-1)\\
    	a_{m}+b_{m}=2(q-1)
    \end{cases}$$
    There exists $\frac{q(5q-3)}{2}$ pairs corresponding $m_i$ and $m_j$ 
    $$
    \begin{cases}
    	m_i&=x_{m-2}^{v-a_{m-1}-{a_m}}x_{m-1}^{a_{m-1}}x_m^{a_m}\\
    	m_j&=x_{m-2}^{v-3(q-1)+a_{m-1}+a_m}x_{m-1}^{q-1-a_{m-1}}x_m^{2(q-1)-a_m}
    \end{cases}
    $$
    From (a)(b)(c),we can conclude that there exists $\frac{q(13q-7)}{2}$ entries with value 1 in the matrix $G'(G')^T$.\\
    \textbf{Case 8}:If $a_{m-3}+b_{m-3}\ne 0,a_{m-2}+b_{m-2}\ne 0,a_{m-1}+b_{m-1}\ne 0$,i.e 
    $$
    \scriptstyle\sum_{p \in \mathcal{P}^{\prime}_m}\left(m_i m_j\right)(p)\\
    		=\left(0^{a_0+b_0} \cdots 0^{a_{m-4}+b_{m-4}}\right) 1^{a_{m-3}+b_{m-3}}\left(\sum_{\alpha_{m-2} \in \mathbb{F}_q} \alpha_{m-2}^{a_{m-2}+b_{m-2}}\right)\left(\sum_{\alpha_{m-1} \in \mathbb{F}_q} \alpha_{m-1}^{a_{m-1}+b_{m-1}}\right)\left(\sum_{\alpha_m \in \mathbb{F}_q} \alpha_m^{a_m+b_m}\right)\\
    $$
    Note that 3(q-1)<2v<4(q-1).If we want to the sumation is nonzero,we must have that 
    $$\begin{cases}
    	a_{m-3}+b_{m-3}=2v-3(q-1)\\
    	a_{m-2}+b_{m-2}=q-1\\
    	a_{m-1}+b_{m-1}=q-1\\
    	a_{m}+b_{m}=q-1
    \end{cases}$$
    ,then $\sum_{p \in \mathcal{P}^{\prime}_m}\left(m_i m_j\right)(p)=-1$.We can get the coresponding pairs $m_i$ and $m_j$ as follow: 
    $$
    \begin{cases}
    	m_i&=x_{m-3}^{v-a_{m-2}-a_{m-1}-a_m}x_{m-2}^{a_{m-2}}x_{m-1}^{a_{m-1}}x_m^{a_m}\\
    	m_j&=x_{m-3}^{v-3(q-1)+a_{m-2}+a_{m-1}+a_m}x_{m-2}^{(q-1)-a_{m-2}}x_{m-1}^{q-1-a_{m-1}}x_m^{(q-1)-a_m}
    \end{cases}
    $$
    We must have that
    $$
    \begin{cases}
    	a_{m-i} \le q-1,(i=0,1,2)\\
    	3(q-1)-v \le a_{m-2}+a_{m-1}+a_m \le v
    \end{cases}    
    $$
    Using\cite{Sor} Lemma 5 ,there exists $\sum_{i=3(q-1)-v}^{v}\sum_{j=0}^{3}(-1)^j\binom{3}{j}\binom{i-qj+2}{i-qj}$ entries in $GG^T$ with value -1 in Case 8. 
    From the analysis of the values and positions of the entries in the matrix $G'(G')^T$,we find that only the third situation of -1 in Case 4 and the Case 8 have contributions to the rank of $G'(G')^T$.We claim that 
    $$
    Rank(G'(G')^T)=\left[\sum_{i=3(q-1)-v}^{v}\sum_{j=0}^{3}(-1)^j\binom{3}{j}\binom{i-qj+2}{i-qj}+[2v-3(q-1)+1]\right]
    $$ . We use the Lemma \ref{lem-parameters-2} and Proposition \ref{prop-hull} and get the Theorem \ref{th_k_to_2(q-1)}.
    
    By excluding the monomials in the third situation of \textbf{Case 4} and \textbf{Case 8}, a basis for $\operatorname{Hull}\left(PRM(q,m,v)\right)$ is given by $\{\operatorname{ev}(f)\}_{f \in \mathcal{M}}$ where $\mathcal{M}$ is the set of monomials of degree $v$ in $x_0, x_1, \ldots, x_m$ except for the $\left[\sum_{i=3(q-1)-v}^{v}\sum_{j=0}^{3}(-1)^j\binom{3}{j}\binom{i-qj+2}{i-qj}+[2v-3(q-1)+1]\right]$ polynomials:
    $$
    \begin{cases}
    	m_i&=x_{m-1}^{v-a_m}x_m^{a_m},(3(q-1)-v \le a_m \le v)\\
    	m_i&=x_{m-3}^{v-a_{m-2}-a_{m-1}-a_m}x_{m-2}^{a_{m-2}}x_{m-1}^{a_{m-1}}x_m^{a_m},\begin{cases}
    		a_{m-\mu} \le q-1,(\mu=0,1,2)\\
    		3(q-1)-v \le a_{m-2}+a_{m-1}+a_m \le v
    	\end{cases}
    	
    \end{cases}.
    $$
\end{proof}
    \begin{example}
    We examine the codes $Hull(PRM(8,4,13))$ by magma:	
    	\begin{lstlisting}
    		%calculate the dimension of Hull(PRM(8,4,13)) directly
    		q:=8;m:=4;v:=13;
    		F := GF(q);
    		P3 := ProjectiveSpace(F, m);
    		points := Points(P3);
    		R<x_1, x_2, x_3,x_4,x_5> := PolynomialRing(F, m+1);
    		polynomials := [R!m : m in SetToSequence(MonomialsOfDegree(R, v)) ];
    		codewords := [];
    		for poly in polynomials do
    		codeword := [];
    		for pt in points do
    		
    		Append(~codeword, Evaluate(poly, Eltseq(pt)));
    		end for;
    		Append(~codewords, Vector(F, codeword));
    		end for;
    		G := Matrix(codewords);
    		linear_code := LinearCode(G);
    		
    		v_mod_q:=v mod (q-1);
    		s_set:={s_1:s_1 in [v_mod_q..v by (q-1)]};
    		dim_prm_qnk:=&+[&+[ (-1)^j*Binomial(m+1, j)*Binomial(t-q*j+m, t-j*q) : j in [0..m+1] ]:t in s_set];
    		dim_prm_qnk-Rank(G*Transpose(G));
    		1682
    		
    		%calculate the dimension of Hull(PRM(8,4,13)) by our formula
    		i_low:=3*(q-1)-v;
    		i_up:=v; 
    		sum_coefficients:=&+[&+[(-1)^j*Binomial(3,j)*Binomial(i-j*q+2,i-j*q) : j in [0..3] ]:i in [i_low..i_up]];
    		dim_prm_qnk-(sum_coefficients+2*v-3*(q-1)+1);
    		1682
    	\end{lstlisting}
    \end{example}
Similar to corollary \ref{cor-k_to_3/2(q-1)} we can get corollary \ref{cor-k_to_2(q-1)}.
\begin{cor}\label{cor-k_to_2(q-1)}
	Suppose that $ m\ge 4, (m-2)(q-1)< v < m(q-1)- \frac{3(q-1)}{2}$ and let $\ell = m(q-1) - v$. Then 
	$$
	\scriptstyle\operatorname{dim}\left(\operatorname{Hull}\left(PRM(q,m,v)\right)\right)=\operatorname{dim}\left(\operatorname{Hull}\left(PRM(q,m,\ell)\right)\right)=k_2-\left[\sum_{i=3(q-1)-\ell}^{\ell}\sum_{j=0}^{3}(-1)^j\binom{3}{j}\binom{i-qj+2}{i-qj}+[2\ell-3(q-1)+1]\right]; .
	$$
	where $k_2=\operatorname{dim}\left(PRM(q,m,\ell)\right)=\sum_{\substack{t=\ell(\bmod q-1) \\
			0<t \leqslant \ell}}\left(\sum_{j=0}^{m+1}(-1)^j\binom{m+1}{j}\binom{t-j q+m}{t-j q}\right).$\\
\end{cor}
 \begin{example}
 	We use magma to examine $Hull(PRM(7,4,13))$,the original magma codes is followed.
 	\begin{lstlisting}
 		%calculate the dimension of Hull(PRM(7,4,13))
 		q:=7;m:=4;v:=13;
 		F := GF(q);
 		P3 := ProjectiveSpace(F, m);
 		points := Points(P3);
 		R<x_1, x_2, x_3,x_4,x_5> := PolynomialRing(F, m+1);
 		polynomials := [R!m : m in SetToSequence(MonomialsOfDegree(R, v)) ];
 		codewords := [];
 		for poly in polynomials do
 		codeword := [];
 		for pt in points do
 		
 		Append(~codeword, Evaluate(poly, Eltseq(pt)));
 		end for;
 		Append(~codewords, Vector(F, codeword));
 		end for;
 		G := Matrix(codewords);
 		linear_code := LinearCode(G);
 		v_mod_q:=v mod (q-1);
 		s_set:={s_1:s_1 in [v_mod_q..v by (q-1)]};
 		dim_prm_qnk:=&+[&+[ (-1)^j*Binomial(m+1, j)*Binomial(t-q*j+m, t-j*q) : j in [0..m+1] ]:t in s_set];
 		dim_prm_qnk-Rank(G*Transpose(G));
 		961
 		
 		%calculate the dimension of Hull(PRM(7,4,13)) by our formula
 		l:=m*(q-1)-v;
 		i_low:=3*(q-1)-l;
 		i_up:=l; 
 		sum_coefficients := &+[&+[ (-1)^j*Binomial(3, j)*Binomial(i-j*q+2, i-j*q) : j in [0..3] ]:i in [i_low..i_up]]+(2*l-3*(q-1)+1);
 		
 		l_mod_q:=l mod (q-1);
 		t_set:={t_1:t_1 in [l_mod_q..l by (q-1)]};
 		dim_prm_qml:=&+[&+[ (-1)^j*Binomial(m+1, j)*Binomial(t-q*j+m, t-j*q) : j in [0..m+1] ]:t in t_set];
 		dim_prm_qml-sum_coefficients;
 		961
 		
 		
 		
 		
 		
 	\end{lstlisting}
 	
 \end{example}   
Using the same approach, though the dimension $m$ of projective space is very large ($m \ge 5$ or more bigger), we can determine the dimension of $Hull(PRM(q,m,v))$ when the range of $v$  gradually increases by increments of $\frac{q-1}{2}$ from $2(q-1)$ or decrease by increments of $\frac{q-1}{2}$ from $(m-2)(q-1)$. However, the analysis of the matrix $G'(G')^T$ may be too complex and hard when $v$ become close to the middle $\frac{m(q-1)}{2}$. So we guess that there doesn't exist a general formula for the dimension of the hulls of PRM codes, or can be studied by other way.  
	\section{The minimum distance of the Hull of $PRM$ code}\label{minihulldistance}
	In this section, we use $wt(PRM(q,m,v))$ denote the minimal distance of a PRM code, or a linear code similiarily.
	\begin{lemma}
		For fixed m and q, the minimal distance of a family of $PRM$ codes $PRM(q,m,v)$ decreases with v increases.
	\end{lemma}
    \begin{proof}
    	This lemma is easy to prove because the minimal distance is $d=(q-s)q^{m-r-1}$, where $v-1=r(q-1)+s,~~0 \le s \le q-1$, so we omit it.
    \end{proof}
\begin{lemma}\label{123subspace}
	For three linear subspace $W_1$, $W_2$, $W_3 \subset \F_q^n$, if  $W_2 \subset W_1$, then 
	$$
	W_1 \cap {{\Spa}}_{\F_q}\{W_3,W_2\} = (W_1 \cap W_3) + (W_1 \cap W_2).
	$$
	In particular, if $W_1 \cap W_3=\{0\}$ , $W_2 \cap W_3=\{0\}$, then 
	$$
	W_1 \cap {{\Spa}}_{\F_q}\{{W_3,W_2}\}=W_1 \cap (W_3 \oplus W_2)=W_1 \cap W_2=W_2.
	$$
\end{lemma}
\begin{proof}
	Generally speaking, for an arbitrary subspace $W_3$, 
	$$W_1 \cap {{\Spa}}_{\F_q}\{W_3,W_2\} \supseteq (W_1 \cap W_3) + (W_1 \cap W_2).$$
	Because for arbitrary $\alpha \in (W_1 \cap W_3) + (W_1 \cap W_2)$, $\alpha =\alpha_{1}+\alpha_{2}, \alpha_{1} \in (W_1 \cap W_3), \alpha_{2} \in (W_1 \cap W_2)$. Then $\alpha \in W_1$ and $\alpha \in {{\Spa}}_{\F_q}\{W_3,W_2\}$. Hence $\alpha \in W_1 \cap {{\Spa}}_{\F_q}\{W_3,W_2\}.$\\
	To prove 
	$$W_1 \cap {{\Spa}}_{\F_q}\{W_3,W_2\} \subseteq (W_1 \cap W_3) + (W_1 \cap W_2),$$
	we can choose an arbitrary element $\alpha$ in $W_1 \cap {{\Spa}}_{\F_q}\{W_3,W_2\}$ and write in the form $\alpha=\beta + \gamma$, where $\beta \in W_3$ and $\gamma \in W_2 \subseteq W_1$. Hence $\beta=\alpha -\gamma \in W_1$. We get that $\beta \in W_1 \cap W_3 = \{\mathbf{0}\}$ and $\beta=0$. Therefore $\alpha \in W_1 \cap W_2 \subseteq (W_1 \cap W_3) + (W_1 \cap W_2).$
	\end{proof}
     
	\begin{theorem} \label{hulldistance}
		Let $v$ be an integer satisfying $1 \le v \le m(q-1)$ and let $\ell =m(q-1)-v$. Then the minimum distance of $wt(PRM(q,m,v))$ satisfies that
		$$
		wt\left(\operatorname{Hull}\left(PRM(q,m,v)\right)\right)= max\{wt(PRM(q,m,v)),wt(PRM(q,m,\ell))\},
		$$
		i.e
		\begin{equation}
			wt\left(\operatorname{Hull}\left(PRM(q,m,v)\right)\right)=
			\begin{cases}
				wt(PRM(q,m,v)) & \text{if $v \le \ell$}\\
				wt(PRM(q,m,\ell)) & \text{if $v \ge \ell.$}
			\end{cases} 
		\end{equation}
	\end{theorem}
\begin{proof}
	By Theorem 1 of \citeonline{Sor},a codeword reach the minimum distance of $PRM(q,m,v)$ corresponds to the following polynomial:
	$$
	F(X)=X_r\prod_{i=0}^{r-1}(X_i^{q-1}-X_r^{q-1})\prod_{j=1}^{s}(\lambda_j X_r-X_{r-1})
    $$
    where $\lambda_i \ne \lambda_j$ for $i \ne j$ and $\lambda_i \in \F_q^{*}$, has zeros at all points except those of the form $(0:0:\dots:1:a_{r+1}:\dots:a_m)$,where $a_r+1 \ne \lambda_j, j=1,\dots,s$ and $a_t \in \F_q^{*}$ for $t=r+2,\dots,m$. Suppose that $PRM(q,m,v)$ and $PRM(q,m,\ell)$ are two PRM codes satisfying the conditions in Theorem \ref{thm-Sor},i.e $v+\ell=m(q-1)$. 
    
    If $v \le l$, the minimum distance of $PRM(q,m,v)$ is bigger than $PRM(q,m,\ell)$, and we notice that we can construct a homogeneous polynomial of degree $l$ as following:
    $$
    G(X)=X_r^{l-k}F(X)=X_r^{l-k+1}\prod_{i=0}^{r-1}(X_i^{q-1}-X_r^{q-1})\prod_{j=1}^{s}(\lambda_j X_r-X_{r-1})
    $$
    Hence $G(X) \in \F_q[X_0,\dots,X_n]^h_l$ and corresponds to the same codeword as $F(X)$. The codeword corresponding to the polynomial  $G(X)$ belongs to $PRM(q,m,\ell)$ which is always contained in ${PRM(q,m,v)}^{\perp}$ for any $v$ satisfying conditions in Theorem \ref{thm-Sor}. Because $G(X)$ and $F(X)$ have the same value in each point of $\mathbb{P}^n(\F_q)$ , the two codewords corresonding to $G(X)$ and $F(X)$ both reach the minimum distance of $PRM(q,m,v)$, i.e they correspond to the same codeword in $PRM(q,m,v)$ which reach the minimal distance of $PRM(q,m,v)$ and the codeword belongs to $\operatorname{Hull}\left(PRM(q,m,v)\right)$. Because $wt(\operatorname{Hull}\left(PRM(q,m,v)\right))\ge wt(PRM(q,m,v))$ by definition of hull. Hence $wt\left(\operatorname{Hull}\left(PRM(q,m,v)\right)\right)=wt\left(PRM(q,m,v)\right)$ when $v \le \ell$. 
    
    If $v \ge \ell$ and $v \not\equiv 0 \pmod{q-1}$ , then $wt\left(\operatorname{Hull}\left(PRM(q,m,v)\right)\right)=wt\left(\operatorname{Hull}\left(PRM(q,m,\ell)\right)\right)=wt\left(PRM(q,m,\ell)\right)$ by theorem \ref{thm-Sor} and the above situation. 
    
    If $v \ge \ell$ and $v \equiv 0 \pmod{q-1}$ , then $\ell \equiv 0 \pmod{q-1}$ and the dual code of $PRM(q,m,v)$ is ${{\Spa}}_{\F_q}  \{ {\bf 1}, PRM(q,m,\ell) \}$ by theorem \ref{thm-Sor}. Hence $\operatorname{Hull}\left(PRM(q,m,v)\right)=\left(PRM(q,m,v)\right) \cap {{\Spa}}_{\F_q}  \{ {\bf 1}, PRM(q,m,\ell) \}$. By lemma \ref{1notinPRM}, $\mathbf{1} \not \in PRM(q,m,v)$ and $\mathbf{1} \not \in PRM(q,m,\ell)$. Hence $PRM(q,m,v) \cap {{\Spa}}_{\F_q}\{ {\bf 1}\}=\{0\}$ , $PRM(q,m,\ell) \cap {{\Spa}}_{\F_q}  \{ {\bf 1}\}=\{0\}$. By lemma \ref{PRMcontain},$PRM(q,m,\ell) \subset PRM(q,m,v).$ In lemma \ref{123subspace}, let $W_1=PRM(q,m,v), W_2=PRM(q,m,\ell), W_3={{\Spa}}_{\F_q}  \{ {\bf 1}\}$, then 
    $$
    \operatorname{Hull}\left(PRM(q,m,v)\right)= PRM(q,m,\ell) 
    $$ 
    by the particular situation. The minimum distance of $PRM(q,m,v)$ is the minimal distance of $PRM(q,m,\ell)$ i.e $wt\left(\operatorname{Hull}\left(PRM(q,m,v)\right)\right)=wt\left(PRM(q,m,\ell)\right)$.
   
\end{proof}
.
\begin{example}
	We use magma to examine the minimal distance of $Hull(PRM(5,3,3))$ and $Hull(PRM(5,3,9))$, and the original codes as follow:
	\begin{lstlisting}
		%calculate the minimal distance of Hull(PRM(5,3,3))
		q:=5;m:=3;v:=3;
		F := GF(q);
		P3 := ProjectiveSpace(F, m);
		points := Points(P3);
		R<x_1, x_2,x_3,x_4> := PolynomialRing(F, m+1);
		polynomials := [R!m : m in SetToSequence(MonomialsOfDegree(R, v)) ];
		codewords := [];
		for poly in polynomials do
		codeword := [];
		for pt in points do
		
		Append(~codeword, Evaluate(poly, Eltseq(pt)));
		end for;
		Append(~codewords, Vector(F, codeword));
		end for;
		G := Matrix(codewords);
		C_1 := LinearCode(G);
		C_2:=Dual(C_1);
		MinimumWeight(C_1 meet C_2);
		75
		%calculate the minimal distance of PRM(5,3,3) by formula in lemma 2.2
		r:=(v-1) div  (q-1);s:=(v-1) mod (q-1);d:=(q-s)*q^(m-r-1);d;
		75
		%calculate the minimal distance of Hull(PRM(5,3,9))
		q:=5;m:=3;v:=9;
		F := GF(q);
		P3 := ProjectiveSpace(F, m);
		points := Points(P3);
		R<x_1, x_2,x_3,x_4> := PolynomialRing(F, m+1);
		polynomials := [R!m : m in SetToSequence(MonomialsOfDegree(R, v)) ];
		codewords := [];
		for poly in polynomials do
		codeword := [];
		for pt in points do
		
		Append(~codeword, Evaluate(poly, Eltseq(pt)));
		end for;
		Append(~codewords, Vector(F, codeword));
		end for;
		G := Matrix(codewords);
		C_1 := LinearCode(G);
		C_2:=Dual(C_1);
		MinimumWeight(C_1 meet C_2);
		75
	\end{lstlisting}
\end{example}

By the above theorem \ref{hulldistance}, for given $q$ and $m$, we can give a low bound for the minimal distance of a PRM code $PRM(q,m,v)$ and derive the following corrollary:

\begin{cor}
	For given $q$ and $m$, let $1 \le v \le \frac{m(q-1)}{2}$, we have that
	$$
	wt(\operatorname{Hull}\left(PRM(q,m,v)\right)) \ge wt(PRM(q,m,\lfloor \frac{m(q-1)}{2} \rfloor))
	$$
	or equivalently,
	$$
	wt(\operatorname{Hull}\left(PRM(q,m,v)\right)) \ge wt(PRM(q,m,\lceil \frac{m(q-1)}{2} \rceil)).
	$$
\end{cor}

In remark \ref{special-case}, we find a special case in which $v \ge \frac{m(q-1)}{2},v \equiv 0 \pmod{q-1}$ i.e $v$ is an even multiple of $\frac{q-1}{2}$, but not any type of Self-dual, Self-orthgonal, LCD or Dual-containing. From the above theorem \ref{hulldistance}, we analyse this case and get that$$
\operatorname{Hull}\left(PRM(q,m,v)\right)= PRM(q,m,\ell) 
$$
Hence we have the following corrollary.
\begin{cor}
	If $v \ge \frac{m(q-1)}{2},v \equiv 0 \pmod{q-1},~~\ell=m(q-1)-v$, the dimension of $Hull(PRM(q,m,v))$ satisfies:
	$$
	\dim(Hull(PRM(q,m,v)))=\dim(PRM(q,m,\ell))=\sum_{\substack{t \equiv \ell\hspace{-.22cm} \pmod{q-1} \\ 0 < t \le \ell}} \left(
	\sum_{j=0}^{m+1} (-1)^j \binom{m+1}{j} \binom{t-jq+m}{t-jq} \right)
	$$
\end{cor}

	\section{Conclusion}
	This paper completely determine the minimal distance for any Projective Reed-Muller Code determined by three parameters $q,m,v$ and analyses two special classes of Projective Reed-Muller Codes omitted by Kaplan and Kim. Although we determine the hull dimension of $PRM(q,m,v)$ when $q-1 < v < \frac{3(q-1)}{2}$ and $\frac{3(q-1)}{2} < v < 2(q-1)$, $(m-1)(q-1)-\frac{q-1}{2} < v < (m-1)(q-1)$and$(m-2)(q-1) < v < (m-1)(q-1)-\frac{q-1}{2}$ as the answer for the question proposed by Kaplan and Kim \cite{kaplankim}, it seems that determining the dimensions under a broader range of parameters is quite challenging, not to mention deriving a general formula. As the parameter range is further expanded, the corresponding matrices become very complex and difficult to analyze, so it may be worth exploring solutions from other perspectives.


\begin{thebibliography}{00}
		
		
		\bibitem{kasamiLinPeterson} T. Kasami, S. Lin and W. W. Peterson, "New generalizations of the Reed-Muller codes—Part I: Primitive codes", IEEE Trans. Inform. Theory, vol. IT-14, no. 2, Mar. 1968.
		
		\bibitem{Weldon} E. J. Weldon, "New generalizations of the Reed-Muller Codes—Part II: Nonprimitive codes", IEEE Trans. Inform. Theory, vol. IT-14, no. 2, Mar. 1968.
		
		\bibitem{kasamiLinPeterson2} T. Kasami, S. Lin and W. W. Peterson, "Polynomial codes", IEEE Trans. Inform. Theory, vol. IT-14, no. 6, Nov. 1968.
		
		\bibitem{DelsarteMGoethalsMacWilliams} P. Delsarte, J. M. Goethals and F. J. MacWilliams, "On generalized Reed-Muller codes and their relatives", Inform. Contr., vol. 16, pp. 403-442, 1970.
		
		\bibitem{SunDingWang}Zhonghua Sun, Cunsheng Ding, Xiaoqiang Wang, "Two Classes of Constacyclic Codes With Variable Parameters [(qm – 1)/r, k, d]", IEEE Transactions on Information Theory, vol.70, no.1, pp.93-114, 2024.
		
		\bibitem{VladutManin}S. G. Vlâduţ and Y. I. Manin, "Linear codes and modular curves", J. Sou. Math., vol. 30, pp. 2611-2643, 1985.
		
		\bibitem{Lachaud1}G. Lachaud, "Projective Reed-Muller codes" in Lect. Notes in Comp. Sci., Berlin:Springer, vol. 311, 1988.
		
		\bibitem{Lachaud2}G. Lachaud, The parameters of projective Reed-Muller codes. Discrete Math. 81(2):217-
		221, 1990.
		
		\bibitem{Sor}  A. B. S\o{}rensen, Projective Reed-Muller codes. IEEE Trans. Inform. Theory, 37(6):1567-1576,
		1991.
		
		\bibitem{Sor2}  A. B. S\o{}rensen, A note on a gap in the proof of the minimum distance for Projective Reed-Muller codes. (2023) 4 pp. \url{https://arxiv.org/abs/2310.03574}.
		
		\bibitem{kaplankim} Nathan Kaplan,Jon-Lark Kim,Hulls of Projective Reed-Muller Codes.(2024) 15 pp. \url{https://arxiv.org/abs/2406.04757}.
		
		\bibitem{Ber} T. P. Berger, Automorphism groups of homogeneous and
		projective Reed-Muller codes, IEEE Trans. Inform. Theory 48 (2002), no. 5, 1035-1045.
		
		\bibitem{BeeDatGho2} P. Beelen, M. Datta, and S. Ghorpade, Maximum number of common zeros of homogeneous polynomials over finite fields.  Proc. Amer. Math. Soc. 146 (2018), no. 4, 1451-1468.
		
		\bibitem{BeeDatGho} P. Beelen, M. Datta, and S. Ghorpade, A combinatorial approach to the number of solutions of systems of homogeneous polynomial equations over finite fields.  Mosc. Math. J. 22 (2022), no. 4, 565-593.
		
		
		
		\bibitem{CheLinLiu} B. Chen, S. Ling, and H. Liu, Hulls of Reed-Solomon codes via algebraic geometry codes. IEEE Trans. Inform. Theory 69 (2023), no. 2, 1005-1014.
		
		
		
		\bibitem{Kim1} S. T. Dougherty, J.-L. Kim, B. Ozkaya, L. Sok, P. Sol\'e, The combinatorics of LCD codes:
		Linear Programming bound and orthogonal matrices, Int. J. Inf. Coding Theory 4 116-128 (2017).
		
		
		\bibitem{Elk} N. D. Elkies, Linear codes and algebraic geometry in higher dimensions. Preprint, 2006.
		
		\bibitem{Roe1} L. Galvez, J.-L. Kim, N. Lee, Y. G. Roe and B. S. Won, Some bounds on binary LCD codes, Cryptogr. Commun. 10  719-728 (2018).
		
		\bibitem{GaoYueHuaZha} Y. Gao, Q. Yue, X. Huang, and J. Zhang, Hulls of generalized Reed-Solomon codes via Goppa codes and their applications to quantum codes. IEEE Trans. Inform. Theory 67 (2021), no. 10, 6619-6626.
		
		\bibitem{GhoLud} S. R. Ghorpade and R. Ludhani, On the minimum distance, minimum weight codewords, and the dimension of projective Reed-Muller codes, Adv. Math. Commun. 18 (2024), no. 2, 360-382.
		
		\bibitem{GueJitGul}
		K. Guenda, S. Jitman, T. A. Gulliver,
		Constructions of good entanglement-assisted quantum error correcting codes,
		Des. Codes and Cryptogr. 86, 121-136 (2018).
		
		\bibitem{Pless1} W. Huffman, V. Pless, Fundamentals of Error-Correcting Codes, Cambridge University Press, Cambridge, 2003.
		
		\bibitem{Kap1} N. Kaplan, Weight enumerators of Reed-Muller codes from cubic
		curves and their duals, 16th International Conference
		``Arithmetic, Geometry, Cryptography, and Coding Theory'', June 19-23, 2017, Contemporary Math., 722, 59-78.
		
		
		
		
		\bibitem{Lac1} G. Lachaud, Projective Reed-Muller codes. In: Coding Theory and Applications (Cachan, 1986). Lecture Notes in Computer Science, vol. 311, pp. 125-129. Springer, Berlin (1988).
		
		
		\bibitem{Lac2}  G. Lachaud, The parameters of projective Reed-Muller codes. Discrete Math. 81(2):217-221, 1990.
		
		
		
		\bibitem{Macwilliams} F. MacWilliams and N. Sloane, The Theory of Error Correcting Codes, North-Holland, London, 1977.
		
		\bibitem{Massey1} J. L. Massey, Reversible codes, Inf. Control 7 (3) 369-380 (1964).
		
		\bibitem{Massey2} J. L. Massey, Linear codes with complementary duals, Discrete Math 106-107 337--342 (1992).
		
		
		
		
		
		\bibitem{Pless2} V. Pless, Introduction to the Theory of Error-Correcting Codes, John Wiley \& Sons, Inc., New York, 1998.
		
		\bibitem{RSJ1} D. Ruano and R. San-Jos\'e, Hulls of projective Reed-Muller codes over the projective plane. (2024) 28 pp. \url{https://arxiv.org/abs/2312.13921}.
		
		\bibitem{RSJ2} D. Ruano and R. San-Jos\'e, The hull variation problem for projective Reed-Muller codes and quantum error-correcting codes. (2024) 12 pp.\\
		\url{https://arxiv.org/abs/2312.15308}.
		
		\bibitem{Euclidean and Hermitian LCD MDS codes}C. Carlet, S. Mesnager, C. Tang and Y. Qi, "Euclidean and Hermitian LCD MDS codes", Designs Codes and Cryptography, 2018.
		\bibitem{New constructions of MDS}B. Chen and H. Liu, "New constructions of MDS codes with complementary duals", IEEE Trans. Inf. Theory, vol. 64, no. 8, pp. 5776-5782, Aug. 2018.
		\bibitem{combinatorics of LCD codes}S. T. Dougherty, J.-L. Kim, B. Özkaya, L. Sok and P. Solé, "The combinatorics of LCD codes: Linear programming bound and orthogonal matrices", Int. J. Inf. Coding Theory, vol. 4, no. 2, pp. 116-128, 2017.
		\bibitem{Esmaeili}M. Esmaeili and S. Yari, "On complementary-dual quasi-cyclic codes", Finite Fields Appl., vol. 15, no. 3, pp. 375-386, Jun. 2009.
		\bibitem{10}C. Güneri, B. Özkaya and P. Solé, "Quasi-cyclic complementary dual codes", Finite Fields Their Appl., vol. 42, pp. 67-80, Nov. 2016.
		\bibitem{13}L. Jin, "Construction of MDS codes with complementary duals", IEEE Trans. Inf. Theory, vol. 63, no. 5, pp. 2843-2847, May 2017.
		\bibitem{16}C. Li, C. Ding and S. Li, "LCD cyclic codes over finite fields", IEEE Trans. Inf. Theory, vol. 63, no. 7, pp. 4344-4356, Jul. 2017.
		
		
		
		
	\end{thebibliography}
\end{document}